\documentclass[letterpaper, 10 pt, conference]{ieeeconf}  

\IEEEoverridecommandlockouts
\usepackage{verbatim}
\usepackage{epsfig}
\usepackage{amsmath}
\usepackage{amsthm}
\usepackage{amssymb}
\usepackage{psfrag}
\usepackage[T1]{fontenc}
\usepackage[ruled]{algorithm2e}
\usepackage{mathtools, nccmath}
\usepackage{cuted}
\usepackage{multirow}
\usepackage{etoolbox}
\AfterEndEnvironment{strip}{\leavevmode}
\usepackage{url}
\usepackage{dsfont}
\usepackage{mathtools}

\usepackage[usenames,dvipsnames]{pstricks}
\usepackage[ruled]{algorithm2e}
\usepackage{etoolbox}
\usepackage{subfig}
\usepackage{graphicx}
\usepackage{lipsum}
\usepackage{bbm}
\usepackage[font=footnotesize,skip=0pt]{caption}
\makeatletter
\patchcmd{\@makecaption}
  {\scshape}
  {}
  {}
  {}
\makeatother
\usepackage[noadjust]{cite}

\newtheorem{theorem}{Theorem}[]
\newtheorem{lemma}{Lemma}[]

\newtheorem{proposition}{Proposition}

\DeclareMathOperator*{\argmax}{argmax}
\DeclareMathOperator*{\argmin}{argmin}

\definecolor{color1}{rgb}{0,0,0}
\usepackage{textcomp}
\usepackage{xcolor}
\def\BibTeX{{\rm B\kern-.05em{\sc i\kern-.025em b}\kern-.08em
    T\kern-.1667em\lower.7ex\hbox{E}\kern-.125emX}}
\usepackage[normalem]{ulem}
\allowdisplaybreaks
\title{\LARGE \bf
Guaranteeing Service in Connected Microgrids: Storage Planning and Optimal Power Sharing Policy}

\author{Arnab Dey$^{1}$, Vivek Khatana$^{1}$, Ankur Mani$^{2}$  and Murti V. Salapaka$^{1}$
\thanks{This work is supported by Advanced Research Projects Agency-Energy OPEN through the project titled "Rapidly Viable Sustained Grid" via grant no. DE-AR0001016.}
\thanks{The authors also acknowledge US Department of Energy for supporting this research through the project titled "CyDERMS: Center for Cybersecurity \& Resiliency of Distribution Energy Resources (DERs) and Microgrids-integrated Distribution Systems" via award no. DE-CR0000040.}
\thanks{$^{1}$ Arnab Dey \{{\tt\small dey00011@umn.edu}\}, Vivek Khatana \{{\tt\small khata010@umn.edu}\}, Murti V. Salapaka\{{\tt\small murtis@umn.edu}\} are with Department of Electrical and Computer Engineering, University of Minnesota, Twin Cities, USA, and $^{2}$ Ankur Mani \{{\tt\small amani@umn.edu\}} is with Department of Department of Industrial and Systems Engineering, University of Minnesota, Twin Cities, USA,
}
}
\begin{document}
\maketitle
\thispagestyle{empty}
\pagestyle{empty}
\begin{abstract}
The integration of renewable energy sources (RES) into power distribution grids poses challenges to system reliability due to the inherent uncertainty in their power production. To address this issue, battery energy sources (BESs) are being increasingly used as a promising solution to counter the uncertainty associated with RES power production. During the overall system planning stage, the optimal capacity of the BES has to be decided. In the operational phase, policies on when to charge the BESs and when to use them to support loads must be determined so that the BES remains within its operating range, avoiding depletion of charge on one hand and remaining within acceptable margins of maximum charge on the other. In this paper, a stochastic control framework is used to determine battery capacity, for microgrids, which ensures that during the operational phase, BESs’ operating range is respected with pre-specified high probability. We provide an explicit analytical expression of the required BESs energy capacity for a single microgrid with RES as the main power source. Leveraging insights from the single microgrid case, the article focuses on the design and planning of BESs for the two-microgrid scenario. In this setting, microgrids are allowed to share power while respecting the capacity constraints imposed by the power lines. We characterize the optimal power transfer policy between the microgrids and the optimal BES capacity for multiple microgrids. This provides the BES savings arising from connecting the microgrids. 

\end{abstract}
\section{Introduction}

In recent years, there has been a tremendous growth in the penetration level of renewable energy sources (RES) in the traditional power distribution networks to address environmental issues, such as carbon emission, and techno-economic issues, such as conventional energy resource depletion, and increasing consumer demand. It is projected that RES will contribute $90$\% of the total energy production by 2050, and approximately 1 trillion USD per year will be invested in the development of renewable energy sectors \cite{irena2019energytransformation}. However, despite the environmental benefits renewable penetration poses a major challenge to the reliable operation of the distribution system as RES are inherently uncertain. Battery energy sources (BESs), provide an attractive option to counter RES uncertainty with the ability to store excess energy when supply exceeds demand and to provide energy when demand exceeds RES-based energy. A key challenge is to determine the capacity of the BES and the associated strategy of when to charge the battery and when to use it to service loads. The problem is difficult because of the stochasticity of RES. Moreover, in the multi-microgrid scenario the facility to share power has the potential to provide benefits of spatial diversity for minimizing needed BES capacity. However, the related problem for determining capacity and charging/discharging policy incorporating the power-sharing ability while respecting constraints of power lines and ensuring BES charge state within specified limits is formidable \cite{hannan2020review}. 
Note that, while BES capacity governs the resource installation cost of the system, the flow capacity of the lines determines the line installation cost. Thus, a proper characterization of both is extremely important to efficiently determine the total installation cost of the system. In this article, we describe a methodology to determine the BES capacity and power-sharing policy respecting the flow capacity.
Summarizing, two tasks posed by RES penetration into the grid are addressed in this article: 

(i) How much BES energy capacity each microgrid (MG) should
have to ensure the BESs are operated within the specified operating range at all times?

(ii) What should be the power-sharing policy amongst the microgrids?


In prior state of the art, the problem of BESs size determination is approached with two primary objectives: (1) operational cost minimization to address techno-economic goals, (2) enhancement of system reliability. Most of the existing works focus on BESs capacity determination with the aim of optimizing installation and operation costs. In this context, several strategies are proposed to find the optimal battery size for both grid-connected and grid-isolated systems \cite{zolfaghari2019optimal,khorramdel2015optimal, li2018optimal}. In \cite{zolfaghari2019optimal,khorramdel2015optimal,chen2011sizing}, battery sizing strategies, based on cost optimization in probabilistic unit commitment settings, are proposed. A stochastic programming approach, to decide the optimal size of battery storage by minimizing the cost of purchasing energy from the grid, is proposed in \cite{hafiz2019energy}. In \cite{li2018optimal}, a chance-constrained approach to minimize the system operation cost, considering a probabilistic constraint on the amount of underproduction, is synthesized. However, overproduction scenarios are not addressed. Furthermore, various optimization techniques \cite{ma2022optimal, khezri2022impact, yang2018battery} are proposed for the optimal allocation of BESs to enable cost-efficient operation of renewable-powered microgrids. However, in critical application scenarios, such as hospitals, reliability of the system is more important compared to cost optimization. For improving system reliability, various techniques in the state-of-the-art include chance-constrained optimization to minimize power curtailment \cite{ xie2022coordinate}, battery charge-discharge control \cite{aghamohammadi2014new}, data-driven approaches \cite{tan2010stochastic}, load-shaping techniques \cite{yang2013sizing}. Further, optimization methods to find optimal energy storage size for load loss minimization are devised \cite{khasanov2021optimal}. However, such methodologies consider either generation or load curtailment probability as the reliability index and thus are difficult to deploy in critical applications where minimizing curtailment of both generation and load is of importance. In \cite{habibi2018allocation}, a numerical analysis with joint optimization of generation and load curtailment is proposed. However, no guarantees of demand fulfillment are provided. Despite the importance, works on determining optimal BESs capacity to minimize both generation and load curtailment without degradation of the BESs, are limited in the existing literature. On the other hand, various data-driven power-sharing strategies to improve multi-microgrid operation are discussed in \cite{cao2021optimal,liu2022consensus,wang2023collaborative}. However, to the best of the author's knowledge, existing works on analytically tractable power-sharing policies, in the case of multiple microgrids, and an analysis of how the transmission line capacity affects the choice of battery capacity is scarce.

To this end, we propose a chance-constrained BESs sizing methodology and power-sharing policy, for multi-microgrid systems, to improve system reliability by minimizing both generation and load curtailment. Our primary contributions are summarized below:

\noindent 1. We propose an analytical method to determine the optimal BES energy capacity to minimize both load and generation curtailment for a finite time horizon. Our proposed method provides probabilistic guarantees on the amount of curtailment and ensures the available energy in the BESs remains within the operating range throughout the time horizon with high probability.

\noindent 2. We extend our BES energy capacity design methodology to a system having two MGs that can share power with each other via a connecting power line. Here, we provide a power transfer policy that achieves the minimum BES energy capacity, required to be installed in each of the MGs, while ensuring the BESs are operated within specified energy limits at all times. We also quantify the BES investment savings arsing from connected operation of two MGs.

The rest of the paper is organized as follows: The system description, problem formulation, and the solution methodology for the single MG system are provided in Section~\ref{sec:syst_desc},~\ref{sec:prob_formulation}, and Section~\ref{sec:solution_methodology_single_mg} respectively. Section~\ref{sec:two_mg_sol_methodology} provides the problem formulation and solution strategy for a connected two-microgrid system. Experimental results are presented in Section~\ref{sec:simu_results} and Section~\ref{sec:conclution} presents the concluding remarks.
\section{System Description}\label{sec:syst_desc}
We consider an MG having local RES, loads, and $N$ number of BESs. The RES and the BESs are utilized to supply the loads. The system architecture is illustrated in Fig.~\ref{fig:system_arch} which shows a microgrid network receiving power from RES and BESs while servicing loads. We denote the power output of the renewables, load demand, and the total power output of the BESs by $P_{g}(t), D_c(t)$, and $b(t)$ respectively. We consider that each BES is identical and has a maximum capacity of $B_{max} \in \mathbb{R}^+$ kWh.
\begin{figure}[!h]
     \centering
     \includegraphics[scale=0.5,trim={0.3cm 7.1cm 22cm 3.8cm},clip]{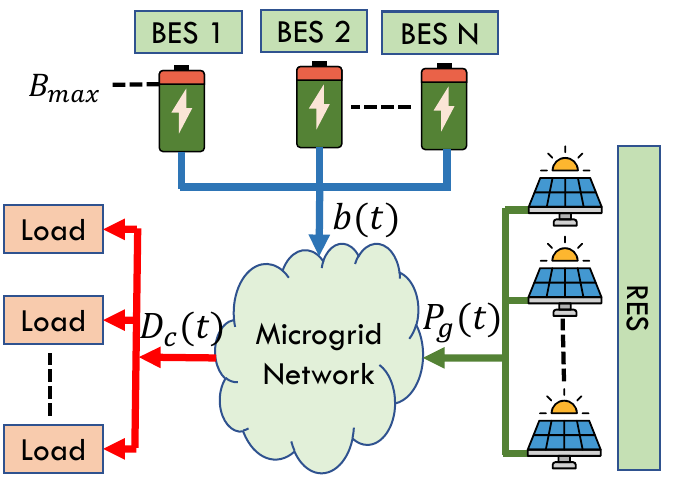}
     \caption{Microgrid architecture}
     \label{fig:system_arch}
\end{figure}
The objective of the microgrid is to determine the number of BESs, $N$, to sustain the system for a given finite time horizon with probabilistic guarantees on the amount of generation and load curtailment by ensuring that the energy available in the BESs remain within acceptable limits. Here we focus on determining optimal BESs capacity which satisfies the following: (i) if the RES produces more power 
than what is demanded by the load, then the battery can be charged with the surplus generation and (ii) when the RES production is less than the required load, the battery has sufficient storage to supply the deficit power. The above conditions are satisfied if the maximum capacity of BESs is not reached and the BESs never get depleted under the uncertainty presented by the RES. In the next section, we make the problem description precise.
\section{Problem Formulation}\label{sec:prob_formulation}
We consider that the load demand (with power $D_c(t)$ demanded at time $t$) in the system is required to be supplied continuously for a time horizon denoted by $T_f$. Note that, the total energy provided by $N$ number of BESs, till time $t\geq 0$ is given by $\int_0^t b(s)\text{d}s = \int_0^t (D_c(s)-P_g(s))\text{d}s$. Further, let the net energy available from the BESs, at any time $t \geq 0$, be denoted by $X(t)$. Then, we have,
\begin{align*}
    \textstyle X(t) = X(0)-\int_0^t b(s)\text{d}s = X(0)+\int_0^t (P_g(s)-D_c(s))\text{d}s,
\end{align*}
where $X(0)$ denotes the initially available energy in the BESs. Further, let $B_0 \in [0, B_{max}]$ kWh at time $t=0$ be the initial charge held by each of the $N$ BESs, and let $\alpha:=B_0/B_{max}$ denote the initial charge to the maximum charge ratio. Thus, the net initial energy available in the BESs is $X(0)=NB_0=\alpha NB_{max}$. The load demand and the RES power generation are both uncertain. To characterize the uncertainty we assume that the net surplus energy, $\int_0^t (P_g(s)-D_c(s))\text{d}s$, follows a Brownian motion characterized by $\sigma W(t)$
where $\sigma\neq 0$ models the volatility, and $W(t)$ is a Weiner process defined on a probability space $(\Omega, \mathcal{F}, \mathbb{P})$.
Thus, $X(t)$ can be written as:
\begin{align}\label{eq:brownian_soc}
    X(t) = \sigma W(t)+\alpha NB_{max}.
\end{align}
Note that, to avoid curtailment of RESs power, the Distribution System Operator (DSO) has to ensure that all the batteries have charge below the designated maximum $B_{\max}$ for all times $t\in [0, T_f]$. Similarly, to reduce load curtailment, BESs should have positive charge for all times $t\in [0,T_f]$. Thus to reduce both generation and load curtailment, the DSO is responsible for sustaining the batteries throughout the time interval $[0,T_f]$ without discharging the batteries completely or charging them to the maximum capacity. Let $X^{sup}_t:= \sup_{s\in [0,t]}X(s)$ and $X^{inf}_t:=\inf_{s\in[0,t]}X(s)$ denote the maximum and minimum energy available in the BESs at any time $s\in [0,t]$ respectively. Therefore, the DSO needs to ensure the following:
\begin{align}\label{eq:prob_objective}
    \mathbb{P}[X^{sup}_{T_f}< NB_{max}, X^{inf}_{T_f}> 0] \geq 1-\delta,
\end{align}
where $\delta>0$ is a small pre-specified tolerance in the probability of violating the conditions on the charge of the BESs. In the above setup, the decision variables include the initial charge $B_0$, and the number $N$ of batteries to be used. The problem translates to: given $\sigma, B_{max},\delta$, decide the minimum number, $N$, of batteries required, and the initial charge ratio $\alpha$, such that~(\ref{eq:prob_objective}) and~\eqref{eq:brownian_soc} are satisfied. Thus the following chance-constrained problem needs to be solved:
\begin{align}\label{eq:problem_statement}
\min_{\alpha, N} & \textstyle \hspace{0.1in} N \\
  & \text{subject to} \ \textstyle N \geq 0,\nonumber\\
  & \textstyle 0 \leq \alpha \leq 1, \nonumber\\
  & \textstyle X(t)=\sigma W(t)+\alpha NB_{max}, \nonumber\\
  & \textstyle \mathbb{P}[X^{sup}_{T_f}< NB_{max}, X^{inf}_{T_f}> 0] \geq 1-\delta. \label{eq:prob_constraint}
\end{align}
In the planning stage, the DSO solves~(\ref{eq:problem_statement}) to determine the required capacity of BESs. In the next section, we describe the solution methodology for the DSO.
\section{Solution Methodology}\label{sec:solution_methodology_single_mg}
In this section, we describe our proposed methodology to solve the problem of the DSO as described in the previous section. We emphasize that obtaining an analytical solution to~\eqref{eq:problem_statement} reamins intractable. Hence, we first analyze the chance-constraint~\eqref{eq:prob_constraint} to derive a tractable approximation of~(\ref{eq:problem_statement}). Further, we derive the corresponding value of $\alpha$ that leads to an explicit solution to the approximated problem.
\subsection{Analysis of the chance constraint}\label{subsec:chance_constraint_analysis}
Note that, $\mathbb{P}[X^{sup}_{T_f}\mspace{-4mu}<\mspace{-4mu} NB_{max}, X^{inf}_{T_f}> 0] = 1-\mathbb{P}[\{X^{sup}_{T_f} \geq NB_{max}\} \cup \{X^{inf}_{T_f} \leq 0\}]$. Therefore, the constraint~(\ref{eq:prob_constraint}) can be equivalently written as $\mathbb{P}[\{X^{sup}_{T_f} \geq NB_{max}\} \cup \{X^{inf}_{T_f} \leq 0\}] \leq \delta$. Further, for $a\in\mathbb{R}$, let 
\begin{align*}
    T(a):=\inf(t\geq 0, X(t)=a)
\end{align*} denote the first time the battery energy level reaches $a$ kWh. Then we have, $X_{T_f}^{sup} \geq NB_{max}$ if and only if $T(NB_{max}) \leq T_f$, and $X^{inf}_{T_f} \leq 0$ if and only if $T(0) \leq T_f$ \cite{karatzas2012brownian}. Thus, the constraint~(\ref{eq:prob_constraint}) can be written as $\mathbb{P}[\{T(NB_{max}) \leq T_f\} \cup \{T(0)\leq T_f\}] \leq \delta$. Consider the following problem:
\begin{align}\label{eq:problem_statement_conservative_1}
\min_{\alpha, N} & \textstyle \hspace{0.1in} N \\
  & \text{subject to} \ \textstyle N \geq 0,\ \textstyle 0 \leq \alpha \leq 1, \nonumber\\
  & \textstyle X(t)=\sigma W(t)+\alpha NB_{max}, \nonumber\\
  & \textstyle \mathbb{P}[T(NB_{max}) \leq T_f]\textstyle + \mathbb{P}[T(0) \leq T_f] \leq \delta. \label{eq:prob_constraint_conservative_1}
\end{align}
Note that a feasible solution to~\eqref{eq:problem_statement_conservative_1} is also a feasible solution to~\eqref{eq:problem_statement}, as $\mathbb{P}[\{T(NB_{max}) \leq T_f\} \cup \{T(0) \leq T_f\}] \leq \mathbb{P}[T(NB_{max})\leq T_f]+\mathbb{P}[T(0)\leq T_f]$. The solution to~\eqref{eq:problem_statement_conservative_1} provides an upper bound of the solution to~\eqref{eq:problem_statement}. In the subsequent development, we derive a tractable upper bound to the chance constraint~(\ref{eq:prob_constraint_conservative_1}), which we further utilize to find a theoretical solution to~\eqref{eq:problem_statement_conservative_1}.

\subsection{Battery Capacity Determination}\label{sec:equal_prod_cons}
In this section, we derive a tractable non-linear approximation to the chance-constraint~(\ref{eq:prob_constraint_conservative_1}) to obtain an analytical solution to~(\ref{eq:problem_statement_conservative_1}). Here, first we provide an upper bound to the chance-constraint~(\ref{eq:prob_constraint_conservative_1}) as presented in the following theorem:
\begin{lemma}\label{lem:gaussian_tail_bound}
    Let
    \begin{align}\label{eq:f_N_alpha}
        \textstyle f(N, \alpha) := e^{-\frac{N^2B_{max}^2(1-\alpha)^2}{2\sigma^2 T_f}} + e^{-\frac{N^2B_{max}^2\alpha^2}{2\sigma^2 T_f}},
    \end{align}
    where $\alpha \in [0,1]$ and $N\geq 0$. Then, we have,
    \begin{align}
        & \textstyle \mathbb{P}[T(NB_{max}) \leq T_f]  \textstyle + \mathbb{P}[T(0) \leq T_f] \leq f(N,\alpha),
    \end{align}
    for all $N$ and $\alpha$.
\end{lemma}
\begin{proof}
    Please see Appendix~\ref{appendix:gaussian_tail_bound}.
\end{proof}
Further note, for any $N\geq 0$, we have $f(N,\alpha) > 1$ if $\alpha\mspace{-4mu}=\mspace{-4mu}1$ or $0$. However, a feasible solution to~(\ref{eq:problem_statement_conservative_1}) requires $f(N,\alpha) \leq 1$ as $\delta \in (0,1)$. Now, consider the following problem:
\begin{align}\label{eq:problem_statement_conservative_2}
\tilde{N}:= \min_{\alpha, N} & \textstyle \hspace{0.1in} N \\
  & \text{subject to} \ \textstyle N \geq 0,\ \textstyle \alpha \in (0,1), \nonumber\\
  & \textstyle f(N, \alpha) \leq \delta. \label{eq:prob_constraint_conservative_2}
\end{align}
From Lemma~\eqref{lem:gaussian_tail_bound}, we derive that any feasible solution to~\eqref{eq:problem_statement_conservative_2} is also a feasible solution to~\eqref{eq:problem_statement_conservative_1}. Thus the solution to~\eqref{eq:problem_statement_conservative_2} provides an upper bound of the solution to~\eqref{eq:problem_statement_conservative_1}. In the subsequent development, we show how an explicit solution to $\tilde{N}$ is obtained. Let $N(\alpha)$ denote the solution to~\eqref{eq:problem_statement_conservative_2} for a given $\alpha\in(0,1)$, that is,
\begin{align}\label{eq:problem_statement_conservative_2_1}
N(\alpha) := \min_{N} & \textstyle \hspace{0.1in} N \\
  & \text{subject to} \ \textstyle N \geq 0,\ \textstyle f(N, \alpha) \leq \delta.\nonumber
\end{align}
Then we have,
\begin{align}\label{eq:problem_statement_conservative_2_2}
\tilde{N} = \min_{\alpha} & \textstyle \hspace{0.1in} N(\alpha) \\
  & \text{subject to} \ \textstyle \alpha \in (0,1) \nonumber.
\end{align}
We present the following theorem to provide an explicit solution to~(\ref{eq:problem_statement_conservative_2}) by considering the sub-problems~\eqref{eq:problem_statement_conservative_2_1} and~\eqref{eq:problem_statement_conservative_2_2}.
\begin{theorem}\label{theorem:solution_conservative_2_simplified}
    Let
    \begin{align}\label{eq:solution_conservative_2_simplified}
        \textstyle N_o:= \frac{\sqrt{8\sigma^2 T_f \ln(2/\delta)}}{B_{max}}.    
    \end{align}
    Then, for $\alpha=\frac{1}{2}$, the optimal value of~\eqref{eq:problem_statement_conservative_2_1} is given by $N(\frac{1}{2})=N_o$. Further, for all $\alpha \in (0,1)\setminus \{\frac{1}{2}\}$, we have, $N(\frac{1}{2})\leq N(\alpha)$. Thus optimal value of~\eqref{eq:problem_statement_conservative_2_2}, or equivalently the optimal value of the problem~\eqref{eq:problem_statement_conservative_2}, is $\tilde{N}=N(\frac{1}{2})=N_o$.
\end{theorem}
\begin{proof}
    Please see Appendix~\ref{appendix:solution_conservative_2_simplified}.
\end{proof}
\noindent Therefore, given $\sigma,B_{max},\delta$ and $T_f$, an optimal strategy is to  initially charge the BESs to half of their maximum capacities, $B_{max}$. Further, the DSO can derive the optimal number of required BESs using~(\ref{eq:solution_conservative_2_simplified}) such that the batteries can sustain the demand throughout the time interval $[0, T_f]$ without getting discharged completely or charged to the maximum capacity with a probability of $1-\delta$. In the next section, we discuss optimality gap and rate optimality of our proposed solution to the problem~(\ref{eq:problem_statement_conservative_2}) compared to the problem~(\ref{eq:problem_statement_conservative_1}).
\subsection{Optimality gap}
Note that,
the probabilities $\mathbb{P}[T(NB_{max})\leq T_f]$ and $\mathbb{P}[T(0)\leq T_f]$ can be explicitly derived as follows (\cite{karatzas2012brownian}):
\begin{align}\label{eq:cdf_upper_lower}
    & \textstyle \mathbb{P}[T(NB_{max}) \textstyle \leq T_f] \textstyle = \frac{2}{\sqrt{2\pi}}\int_{\frac{NB_{max}(1-\alpha)}{\sigma \sqrt{T_f}}}^{\infty} e^{-\frac{1}{2}x^2}\text{d}x,\nonumber\\
    & \textstyle \mathbb{P}[T(0) \textstyle \leq T_f] \textstyle = \frac{2}{\sqrt{2\pi}}\int_{\frac{\alpha NB_{max}}{\sigma \sqrt{T_f}}}^{\infty} e^{-\frac{1}{2}x^2}\text{d}x.
\end{align}
Hence, for any given $\alpha$, $\mathbb{P}[T(NB_{max}) \textstyle \leq T_f]+\mathbb{P}[T(0) \textstyle \leq T_f]$ is monotonic decreasing in $N$. Let $\hat{N}$ denote the optimal value of the problem~\eqref{eq:problem_statement_conservative_1}. Here, from~(\ref{eq:cdf_upper_lower}), we have the following: Given $\delta, B_{max}$ and $T_f$, there exists a constant $\hat{c} \geq 0$, such that if $\hat{N} = \hat{c} \sigma \sqrt{T_f}$, $\mathbb{P}[T(\hat{N}B_{max}) \textstyle \leq T_f]+\mathbb{P}[T(0) \textstyle \leq T_f]\leq \delta$. Note that, from Theorem~\ref{theorem:solution_conservative_2_simplified}, we have $\tilde{N} = \tilde{c} \sigma \sqrt{T_f}$, where $\tilde{c}=\frac{\sqrt{8\ln(2/\delta)}}{B_{max}}$. Hence, our proposed solution to the problem~(\ref{eq:problem_statement_conservative_2}) is optimal up to a constant compared to the solution to~(\ref{eq:problem_statement_conservative_1}), that is $\lim_{T_f\to \infty} \frac{\hat{N}}{\tilde{N}}=c$, where $c\geq 0$ is a constant.

In the next section, we extend our analysis to a system having two microgrids connected to each other and derive the battery capacity for the entire system. Further, we show how the flow capacity of the power line connecting the MGs affects the required battery capacity and provide a design guideline for determining the same.
\section{Battery Sizing and Power Sharing Policy for Connected multi-microgrids}\label{sec:two_mg_sol_methodology}
\subsection{Problem formulation}
In this section, we extend our battery capacity determination methodology, described in Section~\ref{sec:equal_prod_cons}, to an interconnected two-microgrid system, as shown in Fig.~\ref{fig:system_arch_2_mg}. Each MG has its local RES generation and loads and further, they can share power with each other.
\begin{figure}[!h]
     \centering
     \includegraphics[scale=0.35,trim={0.25cm 1.0cm 22cm 0.1cm},clip]{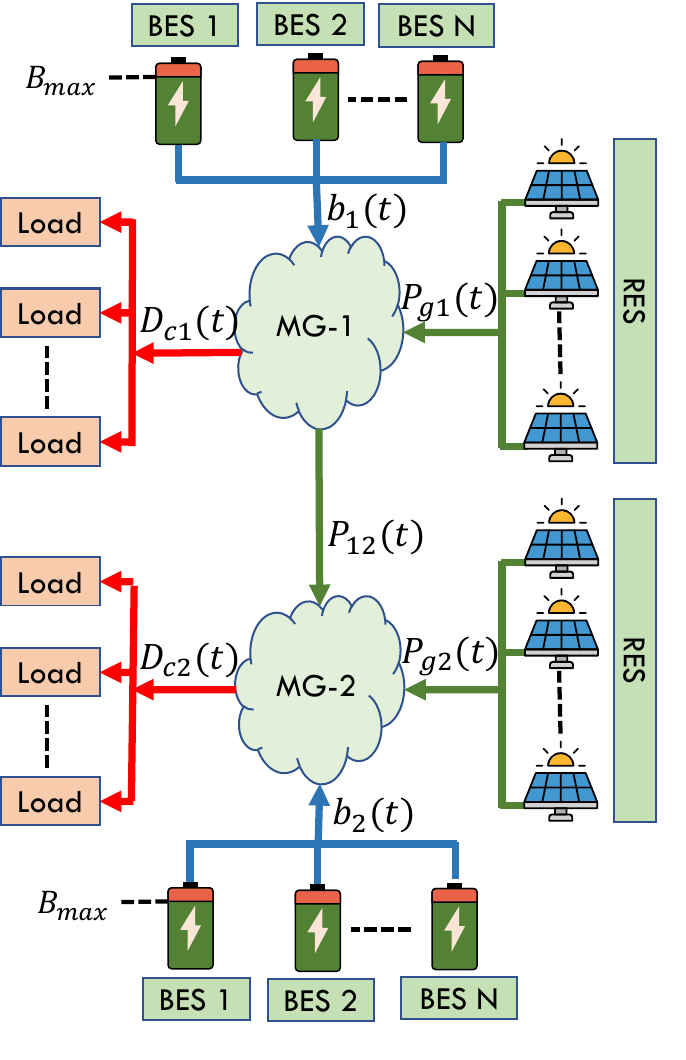}
     \caption{Connected two microgrid architecture}
     \label{fig:system_arch_2_mg}
     \vspace{-0.2cm}
\end{figure}
Let $P_{12}(t)$ denote the amount of power flow from MG-$1$ to MG-$2$ at time $t\in [0,T_f]$. Further, let $P_{gi}(t), D_{ci}(t), b_i(t)$ denote the RES power output, load, and power supplied by the BESs of MG-$i \in \{1,2\}$ respectively. Hence, in each MG, the total energy supplied by $N$ number of BESs, till time $t\geq 0$, is given by:
\begin{align*}
    \textstyle \int_0^t b_1(s)\text{d}s &= \textstyle \int_0^t (D_{c1}(s)-P_{g1}(s))\text{d}s + \int_0^t P_{12}(s)\text{d}s,\\
    \textstyle \int_0^t b_2(s)\text{d}s &= \textstyle \int_0^t (D_{c2}(s)-P_{g2}(s))\text{d}s - \int_0^t P_{12}(s)\text{d}s.
\end{align*}
Denoting $X_i(t), i\in \{1,2\}$ as the net energy available from the BESs of MG-$i, i\in\{1,2\}$ at time $t\geq 0$, we have, $X_i(t) = X_i(0)-\int_0^t b_i(s)\text{d}s$, where $X_i(0)$ denotes the intially available energy in the BESs. Similar to the case of single MG, let $N$ denote the number of BESs to be installed at each MG-$i, i=1,2$, and let $B_0\in[0,B_{max}]$ kWh be the initial charge held by each of the $N$ BESs at time $t=0$. Further, let $\alpha := B_0/B_{max}$ denote the initial charge to the maximum charge ratio. Then, we have, $X_i(0) = \alpha N B_{max}, i \in \{1,2\}$. Similar to the case of a single microgrid, here, we assume
that the net surplus energy, $\int_0^t (P_{gi}(s)-D_{ci}(s))\text{d}s$, follows a Brownian motion characterized by $\sigma W_i(t), t\in [0,T_f]$, where $\sigma \neq 0$ models the volatility, and $W_i(t)$ is a Weiner process. Further assume $W_i(t)$ is independent of $W_j(t), i\neq j, i,j\in\{1,2\}$. Here, we have the following:
\begin{align}\label{eq:brownian_soc_two_mg}
    \textstyle X_1(t) &\mspace{-4mu}=\mspace{-4mu} \textstyle X_1(0)\mspace{-4mu}-\mspace{-4mu}\int_0^t b_1(s)\text{d}s \nonumber\\
    & \textstyle =\mspace{-4mu} \alpha NB_{max} \mspace{-4mu}+\mspace{-4mu} \sigma W_1(t) \mspace{-4mu}-\mspace{-4mu} \int_0^t P_{12}(s)\text{d}s, \text{ and,} \nonumber\\
    \textstyle X_2(t) &\mspace{-4mu}=\mspace{-4mu} \textstyle X_2(0)\mspace{-4mu}-\mspace{-4mu}\int_0^t b_2(s)\text{d}s \nonumber\\
    & \textstyle =\mspace{-4mu} \alpha NB_{max} \mspace{-4mu}+\mspace{-4mu} \sigma W_2(t) \mspace{-4mu}+\mspace{-4mu} \int_0^t P_{12}(s)\text{d}s.
\end{align}
Our problem is divided into two stages; in the planning stage, the DSO determines the number of BESs (or equivalently the BESs capacities), $N$, along with their initial charge ratio $\alpha$. Once the batteries are deployed, the DSO controls the power transfer amount, $P_{12}(t), t\mspace{-4mu}\in\mspace{-4mu}[0,T_f]$, between the MGs in the operational stage. Let the maximum power flow limit of the line, connecting MG-1 and MG-2, is given by $\overline{P}_{12}$. 
Here, we define the power transfer policy, to be determined by the DSO, by $\pi \mspace{-4mu}:\mspace{-4mu}[0,T_f]\mspace{-4mu}\to\mspace{-4mu} \mathbb{R}$ and let $\Pi(\overline{P}_{12})\mspace{-4mu}:=\mspace{-4mu}\{\pi:\mspace{-4mu}[0,T_f]\mspace{-4mu}\to\mspace{-4mu} \mathbb{R} | \pi(t)\mspace{-4mu}\in\mspace{-4mu} [-\overline{P}_{12},\overline{P}_{12}] \mbox{ for all } 0\mspace{-4mu}\leq\mspace{-4mu} t\mspace{-4mu}\leq\mspace{-4mu} T_f\}$ denote the set of permissible policies adhering to the given power line capacity. In the subsequent analysis, we suppress the explicit dependency of $\pi(t)$ on $t$ for the brevity of notation.

Let $X_{i,T_f}^{sup}\mspace{-5mu}:=\mspace{-5mu}\sup_{t\in[0,T_f]}X_i(t)$ and $X_{i,T_f}^{inf}\mspace{-5mu}:=\mspace{-5mu} \inf_{t\in[0,T_f]}X_i(t), i\mspace{-4mu}=\mspace{-4mu}1,2$ denote the maximum and minimum energy available, respectively, at any time $t\mspace{-4mu}\in\mspace{-4mu}[0,T_f]$ in the BESs of MG-$i$. Note that, the available energy in the BESs at any time depends on the number $N$, initial charge ratio $\alpha$, and the power transfer policy $\pi$. Hence, let the probability that the BESs, given an initial charge of $X_i(0) \mspace{-4mu}=\mspace{-4mu} \alpha N B_{max}$, do not get fully charged or emptied any time during the time horizon $[0,T_f]$ be denoted by 
\begin{align*}
    \textstyle P(N,\alpha,\pi) &:= \textstyle \mathbb{P}[X_{1,T_f}^{sup}\mspace{-4mu}<\mspace{-4mu} NB_{max}, \textstyle X_{1,T_f}^{inf}\mspace{-4mu}>\mspace{-4mu} 0,\nonumber\\
    &\textstyle \textstyle X_{2,T_f}^{sup}\mspace{-4mu}<\mspace{-4mu} NB_{max}, X_{2,T_f}^{inf}\mspace{-4mu}>\mspace{-4mu} 0].
\end{align*}
Here, the objective of the DSO is to determine a suitable control strategy $\pi$ and initial charge ratio $\alpha$, that minimizes the number of BESs, $N$, and ensures with a given high probability of $1-\delta$, that the BESs are not fully charged or empty during the time horizon $[0,T_f]$. Therefore, DSO solves the following chance-constrained problem:
Given $B_{max}, \overline{P}_{12}, \delta, T_f$,
\begin{align}\label{eq:problem_statement_two_mg}
\min_{\alpha, N, \pi} & \textstyle \hspace{0.1in} N \\
  & \hspace{-1cm}\text{subject to} \ (\ref{eq:brownian_soc_two_mg}),\ \textstyle N \geq 0,\ \textstyle 0 \leq \alpha \leq 1, \nonumber\\
  & \hspace{-1cm}\textstyle \pi \in \Pi(\overline{P}_{12}),\ \textstyle P(N,\alpha,\pi) \geq\mspace{-4mu} 1\mspace{-4mu}-\mspace{-4mu}\delta.\nonumber
\end{align}
In the planning stage, the DSO solves~(\ref{eq:problem_statement_two_mg}) to decide the number of BESs to be installed in both the MGs, their initial charges, and the power transfer policy. The power transfer policy is followed throughout the time horizon $[0,T_f]$ in the operational stage, once the BESs are deployed. In Section~\ref{subsec:solution_methodology_two_mg}, we present an analytical solution to~(\ref{eq:problem_statement_two_mg}). Further note that, the power line capacity $\overline{P}_{12}$ affects the solution to~(\ref{eq:problem_statement_two_mg}); In the following sections, we characterize such dependency of the optimal solution of~(\ref{eq:problem_statement_two_mg}),  characterizing the number of batteries required, on the power line capacity.

\subsection{Solution methodology}\label{subsec:solution_methodology_two_mg}
Given $B_{max}, \overline{P}_{12}, \delta, T_f$,
let the solution set of~\eqref{eq:problem_statement_two_mg} be denoted by,
\begin{align*}
    \textstyle S_1 &:= \textstyle \argmin\limits_{N,\alpha,\pi} \{N:(N,\alpha,\pi) \in F_1\}, \text{ where,}\\
    \textstyle F_1 &:= \textstyle \{(N,\alpha,\pi):\eqref{eq:brownian_soc_two_mg} \text{ is satisfied}, 0\leq \alpha\leq 1, N\geq 0,\\
    & \pi\in\Pi(\overline{P}_{12}), \textstyle P(N,\alpha,\pi) \geq 1-\delta\}.
\end{align*}
Suppose $(N^*,\alpha^*,\pi^*)$ solves the problem~\eqref{eq:problem_statement_two_mg}, that is $(N^*,\alpha^*,\pi^*) \mspace{-4mu}\in\mspace{-4mu}S_1$ and $N^* \mspace{-4mu}=\mspace{-4mu} \min_{N,\alpha,\pi} \{N:(N,\alpha,\pi) \mspace{-4mu}\in\mspace{-4mu} F_1\}$.
In the following theorem, we show how an equivalent tractable problem can be constructed to obtain $N^*$, the optimal value of~\eqref{eq:problem_statement_two_mg}.
\begin{theorem}\label{theorem:problem_statement_two_mg_dual}
    Consider the following problem:
    \begin{align}\label{eq:problem_statement_two_mg_dual}
    \nu &:= \textstyle \max\limits_{N,\pi} \{P(N,\alpha^*, \pi): (N,\pi) \in F_2\}, \text{ where,}\\
    F_2 &:= \textstyle \{(N,\pi):\eqref{eq:brownian_soc_two_mg} \text{ with } \alpha=\alpha^* \text{ is satisfied},\nonumber\\
    &\textstyle \pi\in\Pi(\overline{P}_{12}),N \leq N^*\}\nonumber,
\end{align}
and let the solution set of~\eqref{eq:problem_statement_two_mg_dual} be denoted by,
\begin{align*}
    \textstyle S_2 := \argmax\limits_{N,\pi} \{P(N,\alpha^*,\pi):(N,\pi) \in F_2\}.
\end{align*}
Suppose $(N',\pi')\in S_2$, then $N'\mspace{-4mu}=\mspace{-4mu}N^*$ and $(N',\alpha^*,\pi') \in S_1$.
\end{theorem}
\begin{proof}
    Please see Appendix~\ref{appendix:problem_statement_two_mg_dual}.
\end{proof}
{\it Remark:}  Thus if $(N',\pi')$ solves problem~\eqref{eq:problem_statement_two_mg_dual} then $(N',\alpha^*,\pi')$ also solves problem~\eqref{eq:problem_statement_two_mg} and $N'$ will be equal to $N^*$ which is the optimal value of~\eqref{eq:problem_statement_two_mg}.

Note that, as $(N',\alpha^*,\pi')$ obtained from the solution to~\eqref{eq:problem_statement_two_mg_dual}, solves the problem~\eqref{eq:problem_statement_two_mg} as well, we first seek to derive $\pi'$ from~\eqref{eq:problem_statement_two_mg_dual}. In this context, let $\tilde{\pi}(N,\alpha)$ denote a power transfer policy that maximizes $P(N,\alpha,\pi)$ for a given $N>0$ and $\alpha\in[0,1]$, that is,
\begin{align}\label{eq:optimal_policy_given_n_alpha_two_mg}
    \tilde{\pi}(N,\alpha) := & \argmax_{\pi}\  P(N,\alpha, \pi) \\
    & \text{subject to } (\ref{eq:brownian_soc_two_mg}), \pi\in\Pi(\overline{P}_{12}). \nonumber
\end{align}
In the subsequent section, we derive the policy $\tilde{\pi}(N,\alpha)$ and show that the policy does not depend on $N$ and $\alpha$ and rather depends on the available energy in the batteries of the MGs.
\subsubsection{Power Transfer Policy Determination}
Let the time horizon $[0,T_f]$ be partitioned into $K$ intervals $0\mspace{-4mu}=\mspace{-4mu}t_0\mspace{-4mu}<\mspace{-4mu}t_1\mspace{-4mu}<\mspace{-4mu}\ldots\mspace{-4mu}<\mspace{-4mu}t_K\mspace{-4mu}=\mspace{-4mu}T_f$, such that $t_k\mspace{-4mu}=\mspace{-4mu}k\Delta t, k\mspace{-4mu}\in\mspace{-4mu}\{0,1,\ldots,K\}$ where $\Delta t = T_f/K$.
Let $\overline{X}_k\mspace{-4mu}:=\mspace{-4mu}[X_1(t_k)\ X_2(t_k)]^T$ denote the vector of available energy in the batteries of the MGs at time $t_k$. Note that, in discrete time, the power transfer policy $\pi$ is a sequence of $K$ power transfer control laws, that is, $\pi \mspace{-4mu}=\mspace{-4mu}[P_{12}(0),P_{12}(t_1),\ldots,P_{12}(t_{K-1})]$. Further, let $P^{\pi}_0(\overline{x}_0)\mspace{-4mu}:=\mspace{-4mu}\mathbb{P}[\overline{X}_k \mspace{-4mu}\in\mspace{-4mu}B, \text{ for all }k\mspace{-4mu}\in\mspace{-4mu}\{1,\ldots,K\}|\overline{X}_0=\overline{x}_0]$ denote the probability that $\overline{X}_k$ remains within some given interval $B\mspace{-4mu}\in\mspace{-4mu}\mathbb{R}^2$ for all $k$, starting with an initial value of $\overline{X}_0=\overline{x}_0$ and following a given policy $\pi$ thereafter. We define $B \mspace{-4mu}:=\mspace{-4mu} [0, NB_{max}]\mspace{-4mu}\times\mspace{-4mu}[0, NB_{max}]$. We choose $\Delta t$ in such a way that $NB_{max}-\overline{P}_{12}\Delta t> 0$. Here, we define $N_B\mspace{-4mu}:=\mspace{-4mu}NB_{max}$, for brevity of notation.
Consider, the problem, given $N_B,\overline{x}_0\in [0,N_B],T_f,\overline{P}_{12}$,
\begin{align}\label{eq:problem_statement_two_mg_dual_discrete}
\max_{\pi} & \textstyle \hspace{0.1in} P^{\pi}_0(\overline{x}_0) \\
  & \hspace{-1cm}\text{subject to} \ (\ref{eq:brownian_soc_two_mg}), \pi\in\Pi(\overline{P}_{12}).\nonumber
\end{align}
Let $\pi'$ solve the problem~\eqref{eq:problem_statement_two_mg_dual_discrete}, that is $\pi'\mspace{-4mu}=\mspace{-4mu}[P'_{12}(t_0), P'_{12}(t_1),\ldots,P'_{12}(t_{K-1})] \mspace{-4mu}=\mspace{-4mu} \argmax_{\pi} \{P^{\pi}_0(\overline{x}_0):\eqref{eq:brownian_soc_two_mg} \text{ is satisfied},\pi\mspace{-4mu}\in\mspace{-4mu}\Pi(\overline{P}_{12})\}$. We utilize the dynamic programming principle to derive $\pi'$.

\noindent Consider the last time step $t_{K-1}=T_f-\Delta t$, and suppose $\overline{X}_{K-1}\mspace{-4mu}=\mspace{-4mu}\overline{x}\mspace{-4mu}=\mspace{-4mu}[x_1,x_2]^T$. From~\eqref{eq:brownian_soc_two_mg} and by the definition of Wiener process, for the last time step, we have, $\overline{X}_K\mspace{-4mu}\sim\mspace{-4mu} \mathcal{N}(\overline{x}+(P_{12}(t_{K-1})\Delta t)[-1\ 1]^T,(\sigma^2 \Delta t)\mathbb{I})$, where $\mathbb{I}\mspace{-4mu}\in\mspace{-4mu}\mathbb{R}^2$ is an identity matrix. Note, $\mathbb{P}[\overline{X}_K\mspace{-4mu}\in\mspace{-4mu}B|\overline{X}_{K-1}\mspace{-4mu}=\mspace{-4mu}\overline{x}\mspace{-4mu}\in\mspace{-4mu}B]$ can further be written as
\begin{align*}
    & \textstyle \mathbb{P}[\overline{X}_K\mspace{-4mu}\in\mspace{-4mu}B|\overline{X}_{K-1}\mspace{-4mu}=\mspace{-4mu}\overline{x}\mspace{-4mu}\in\mspace{-4mu}B]\\
    & \textstyle = \mathbb{P}[X_1(t_K)\mspace{-4mu}\in\mspace{-4mu} [0,N_B]|X_1(t_{K-1})\mspace{-4mu}=\mspace{-4mu}x_1\mspace{-4mu}\in\mspace{-4mu}[0,N_B]]\\
    & \textstyle \mathbb{P}[X_2(t_K)\mspace{-4mu}\in\mspace{-4mu} [0,N_B]|X_2(t_{K-1})\mspace{-4mu}=\mspace{-4mu}x_2\mspace{-4mu}\in\mspace{-4mu}[0, N_B]]\\
    & =\textstyle \Big[\Phi\Big(\frac{N_B-x_1+P_{12}\Delta t}{\sigma \sqrt{\Delta t}}\Big)\mspace{-4mu}-\mspace{-4mu}\Phi\Big(\frac{-x_1+P_{12}\Delta t}{\sigma \sqrt{\Delta t}}\Big)\Big]\\
    & \textstyle\Big[\Phi\Big(\frac{N_B-x_2-P_{12}\Delta t}{\sigma \sqrt{\Delta t}}\Big)\mspace{-4mu}-\mspace{-4mu}\Phi\Big(\frac{-x_2-P_{12}\Delta t}{\sigma \sqrt{\Delta t}}\Big)\Big]\\
    & =\textstyle \Big[\Phi(g_1(x_1,P_{12}))-\Phi(g_2(x_1,P_{12}))\Big]\\
    & \textstyle \Big[\Phi(g_3(x_1,P_{12}))-\Phi(g_4(x_1,P_{12}))\Big],
\end{align*}
where $\Phi(\cdot)$ is the CDF of standard normal, and,
\begin{align*}
    \textstyle g_1(x_1,P_{12})\mspace{-4mu} & \textstyle :=\mspace{-4mu}\frac{N_B-x_1+P_{12}\Delta t}{\sigma \sqrt{\Delta t}},\ 
    \textstyle g_2(x_1,P_{12})\mspace{-4mu} \textstyle :=\mspace{-4mu}\frac{-x_1+P_{12}\Delta t}{\sigma \sqrt{\Delta t}}\\
    \textstyle g_3(x_2,P_{12})\mspace{-4mu} & \textstyle :=\mspace{-4mu}\frac{N_B-x_2-P_{12}\Delta t}{\sigma \sqrt{\Delta t}},\ \textstyle g_4(x_2,P_{12})\mspace{-4mu} \textstyle :=\mspace{-4mu}\frac{-x_2-P_{12}\Delta t}{\sigma \sqrt{\Delta t}},
\end{align*}
and $\overline{x} \in [0, N_B]\times[0, N_B]$.


At the last time step, we need to determine the power transfer amount $P_{12}(t_{K-1})$ that maximizes the probability, $P^{\pi}_{K-1}(\overline{x})=\mathbb{P}[\overline{X}_K\mspace{-4mu}\in\mspace{-4mu}B|\overline{X}_{K-1}\mspace{-4mu}=\mspace{-4mu}\overline{x}\mspace{-4mu}\in\mspace{-4mu}B]$. In the subsequent paragraphs, we omit the explicit dependency of $P_{12}(t_{K-1})$ on time for brevity of notation. 
We further define $V_k^{\pi}(\overline{x}) := P_k^{\pi}(\overline{x})=\mathbb{P}[\overline{X}_k\in B, \text{ for all } k \in \{k+1,\ldots,K\}|\overline{X}_k=\overline{x}_0\in B]$. Thus,
\begin{align*}
    \textstyle V_{K-1}^{\pi}(\overline{x}) &:= \textstyle \Big[\Phi(g_1(x_1,P_{12}))-\Phi(g_2(x_1,P_{12}))\Big]\\
    & \textstyle \Big[\Phi(g_3(x_2,P_{12}))-\Phi(g_4(x_2,P_{12}))\Big].
\end{align*}
Thus, at the last time step, we solve the following problem: Given $\overline{x}\in [0,N_B]$,
\begin{align}\label{eq:problem_two_mg_last_step}
    \sup_{P_{12}} & \textstyle \hspace{0.1in} \textstyle V_{K-1}^{\pi}(\overline{x})\\
    & \hspace{-1cm}\text{subject to} \ -\overline{P}_{12}\leq P_{12}\leq \overline{P}_{12},
\end{align}
In the subsequent paragraphs, we derive a solution to~\eqref{eq:problem_two_mg_last_step}.

\noindent Note that, the Lagrangian of the problem~(\ref{eq:problem_two_mg_last_step}) can be written as follows:
\begin{align}\label{eq:lagrangian_two_mg}
    & \textstyle L(P_{12},\lambda_1,\lambda_2) := V_{K-1}^{\pi}(\overline{x}) \nonumber\\
    & \textstyle -\lambda_1(P_{12}-\overline{P}_{12}) - \lambda_2(-\overline{P}_{12}-P_{12}),
\end{align}
where $\lambda_1\mspace{-4mu}\geq \mspace{-4mu}0$ and $\lambda_2\mspace{-4mu}\geq \mspace{-4mu}0$ are the Lagrange multipliers. The KKT conditions \cite{boyd2004convex} for the Lagrangian $L(P_{12},\lambda_1,\lambda_2)$ can be written as:
\begin{align}\label{eq:problem_two_mg_last_step_kkt}
    \textstyle \frac{\partial L(\cdot)}{\partial P_{12}} = 0,\ \lambda_1(P_{12}\mspace{-4mu}-\mspace{-4mu}\overline{P}_{12})=0, \lambda_2(-\overline{P}_{12}\mspace{-4mu}-\mspace{-4mu}P_{12})=0.
\end{align}
Here, we have the following:
\begin{align}\label{eq:kkt_two_mg}
    & \textstyle \frac{\partial L(\cdot)}{\partial P_{12}} = 0\nonumber\\
    & \textstyle \implies -\frac{\sqrt{\Delta t}}{\sqrt{2\pi}\sigma}\Big[e^{-\frac{1}{2}(g_2(x_1,P_{12}))^2}-e^{-\frac{1}{2}(g_1(x_1,P_{12}))^2}\Big]\nonumber\\
    & \textstyle \Big[\Phi(g_3(x_1,P_{12}))-\Phi(g_4(x_1,P_{12}))\Big]\nonumber\\
    & \textstyle + \frac{\sqrt{\Delta t}}{\sqrt{2\pi}\sigma}\Big[\Phi(g_1(x_1,P_{12}))-\Phi(g_2(x_1,P_{12}))\Big]\nonumber\\
    & \textstyle \Big[e^{-\frac{1}{2}(g_4(x_2,P_{12}))^2}-e^{-\frac{1}{2}(g_3(x_2,P_{12}))^2}\Big] - \lambda_1 + \lambda_2=0.
\end{align}
Further, from the complementary slackness conditions, we can deduce the following:

\noindent Case 1: It is not feasible to have non-zero values for both $\lambda_1$ and $\lambda_2$.

\noindent Case 2: If $\lambda_1 \mspace{-4mu}\neq\mspace{-4mu}0$, then $P_{12}\mspace{-4mu}=\mspace{-4mu}\overline{P}_{12}$ and $\lambda_2\mspace{-4mu}=\mspace{-4mu}0$.

\noindent Case 3: If $\lambda_2\mspace{-4mu}\neq\mspace{-4mu}0$, then $P_{12}\mspace{-4mu}=\mspace{-4mu}-\overline{P}_{12}$ and $\lambda_1\mspace{-4mu}=\mspace{-4mu}0$.

Considering Cases 2, and 3, it can be shown that, (i) if $x_1>x_2+2\overline{P}_{12}\Delta t$, the KKT conditions~\eqref{eq:problem_two_mg_last_step_kkt} are satisfied with $P_{12}=\overline{P}_{12}$, (ii) if $x_2>x_1+2\overline{P}_{12}\Delta t$, the KKT conditions~\eqref{eq:problem_two_mg_last_step_kkt} are satisfied with $P_{12}=-\overline{P}_{12}$, and (iii) if $x_1=x_2$, the KKT conditions are satisfied with $P_{12}=0$.
In the case of, $|x_1-x_2|\mspace{-4mu}\in\mspace{-4mu}(0,2\overline{P}_{12}\Delta t]$, we consider $P_{12}=\frac{x_1-x_2}{2\Delta t}$. Note that, in this case, $\lambda_1=\lambda_2=0$, and, $P_{12}=\frac{x_1-x_2}{2\Delta t}$ satisfies the KKT conditions~\eqref{eq:kkt_two_mg}. Later we show that as $\Delta t$ becomes smaller ($\Delta t \to 0$), the region $|x_1-x_2| \in (0,2\overline{P}_{12}\Delta t]$ shrinks and we obtain a bang-bang power transfer policy.

Hence, in summary, we have the following: a power transfer policy, characterized by,
\begin{align}\label{eq:optimal_policy_delta_t}
    & \textstyle P_{12}=\overline{P}_{12}, \text{ if } x_1>x_2+2\overline{P}_{12}\Delta t,\nonumber\\
    & \textstyle P_{12} = -\overline{P}_{12}, \text { if } x_2>x_1+2\overline{P}_{12}\Delta t,\nonumber\\
    & \textstyle P_{12}=0, \text{ if }x_1=x_2,\nonumber\\
    & \textstyle P_{12}=\frac{x_1-x_2}{2\Delta t}, \text{ if }|x_1-x_2|\in (0, 2\overline{P}_{12}\Delta t],
\end{align}
satisfies the KKT conditions~\eqref{eq:kkt_two_mg}.

\noindent Let the optimal value of~\eqref{eq:problem_two_mg_last_step} be denoted by $V_{K-1}^*(\overline{x}) := \sup_{P_{12}}\{V_{K-1}^{\pi}(\overline{x}):\overline{P}_{12}\leq P_{12}<\overline{P}_{12}\}$. Here we assume that $V_{K-1}^*(\overline{x})$ is concave in $\overline{x} \in [0,N_B]$, which can be empirically verified for the power transfer policy~\eqref{eq:optimal_policy_delta_t}. In the subsequent analysis, we show that under the assumption, $V_{K-1}^*(\overline{x})$ is concave in $\overline{x}$, we have $V_{k}^*(\overline{x})$ is concave in $\overline{x}\in [0,N_B]$ for all $k\in \{0,1,\ldots,K-1\}$, and thus, the power transfer policy~\eqref{eq:optimal_policy_delta_t} is optimal for all $k$.

\noindent Note that, $P_0^{\pi}(\overline{x}_0)$ can be expressed as follows:
\begin{align*}
    \textstyle P_0^{\pi}(\overline{x}_0) = \mathbb{E}[\Pi_{j=1}^{K} \mathbbm{1}_{B}(\overline{X}_j)|\overline{X}_0=\overline{x}_0],
\end{align*}
where $\mathbbm{1}_A(\overline{X}) \mspace{-4mu}=\mspace{-4mu} 1$,  if $\overline{X}\in A$, and $0$ otherwise.

Thus we have,
\begin{align}\label{eq:value_func_definition_two_mg}
    \textstyle V_k^{\pi}(\overline{x}) &= \textstyle \mathbb{E}[\Pi_{j=k+1}^{K}\mathbbm{1}_{B}(\overline{X}_j)|\overline{X}_k = \overline{x}],\\
    \textstyle V_k^*(\overline{x}) & = \textstyle \sup_{\pi}V_k^{\pi}(\overline{x}).
\end{align}

We denote the transition probability by $Q_k(\cdot)$ which is defined as follows:
\begin{align*}
    \textstyle Q_k(\mathcal{X}|\overline{x},P_{12}(t_{k})) := \mathbb{P}[\overline{X}_{k+1}\in \mathcal{X}|\overline{X}_{k}=\overline{x},P_{12}(t_{k})].
\end{align*}
In subsequent analysis, we show that the following recursion,
\begin{align}\label{eq:recusrion_two_mg}
    \textstyle V_K^*(\overline{x}) &= \textstyle 1,\\
    \textstyle V_k^*(\overline{x}) &= \textstyle \sup\limits_{P_{12}(t_k)} \int \mathbbm{1}_{B}(\overline{y})V^*_{k+1}(\overline{y})Q_{k}(\text{d}\overline{y}|\overline{x},P_{12}(t_k)),\nonumber\\
    & \textstyle k\in\{0,1,\ldots,K-1\},\nonumber
\end{align}
solves the problem~\eqref{eq:problem_statement_two_mg_dual_discrete}.

\noindent We start with the base case $k=K-1$. From the definition~\eqref{eq:value_func_definition_two_mg}, we have,
\begin{align*}
    & \textstyle V^*_{K-1}(\overline{x}) \textstyle = \sup\limits_{P_{12}(t_{K-1})}\mathbb{E}[\mathbbm{1}_{B}(\overline{X}_K)|\overline{X}_{K-1}=\overline{x}]\\
    &= \textstyle \sup\limits_{P_{12}(t_{K-1})}\int \mathbbm{1}_{B}(\overline{y})\mathbb{P}[\overline{X}_K\in \text{d}\overline{y}|\overline{X}_{K-1}=\overline{x},P_{12}(t_{K-1})]\\
    &= \textstyle \sup\limits_{P_{12}(t_{K-1})}\int \mathbbm{1}_{B}(\overline{y})V^*_K(\overline{y})Q_{K-1}(\text{d}\overline{y}|\overline{x},P_{12}(t_{K-1})).
\end{align*}
Thus the recursion~\eqref{eq:recusrion_two_mg} holds for $k=K-1$. Now let $k\in\{0,1,\ldots,K-2\}$ and presume that $V^*_{k+1}(\overline{x})$ is known. From the definition~\eqref{eq:value_func_definition_two_mg}, we have the following:
\begin{align*}
    & \textstyle V_k^*(\overline{x}) = \sup\limits_{P_{12}(t_k),\ldots,P_{12}(t_{K-2})}\mathbb{E}[\Pi_{j=k+1}^{K-2}\mathbbm{1}_{B}(\overline{X}_j)|\overline{X}_k=\overline{x}]\\
    & \textstyle = \sup\limits_{P_{12}(t_k),\ldots,P_{12}(t_{K-2})} [\int \Pi_{j=k+1}^{K-2}\mathbbm{1}_{B}(\overline{X}_j)\\
    & \textstyle \Pi_{i=k+2}^{K-2}Q_{i-1}(\text{d}\overline{y}_i|\overline{y}_{i-1},P_{12}(t_{i-1}))Q_{k}(\text{d}\overline{y}_{k+1}|\overline{x},P_{12}(t_k))\\
    & \textstyle =\sup\limits_{P_{12}(t_k)}[\int \mathbbm{1}_{B}(\overline{y}_{k+1})\\
    & \textstyle \Big(\sup\limits_{P_{12}(t_{k+1}),\ldots,P_{12}(t_{K-2})}[\int \Pi_{j=k+2}^{K-2}\mathbbm{1}_{B}(\overline{y}_j)\\
    & \textstyle \Pi_{i=k+2}^{K-2}Q_{i-1}(\text{d}\overline{y}_i|\overline{y}_{i-1},P_{12}(t_{i-1})\Big)Q_{k}(\text{d}\overline{y}_{k+1}|\overline{x},P_{12}(t_{k}))\\
    & = \textstyle \sup\limits_{P_{12}(t_k)}[\int \mathbbm{1}_{B}(\overline{y}_{k+1})V^*_{k+1}(\overline{y}_{k+1})Q_{k}(\text{d}\overline{y}_{k+1}|\overline{x},P_{12}(t_{k}))].
\end{align*}
This completes the proof that~\eqref{eq:recusrion_two_mg} solves the problem~\eqref{eq:problem_statement_two_mg_dual_discrete}.

Under the assumption that $V^*_{K-1}(\overline{x})$ is concave in $\overline{x}$, by the recursion~\eqref{eq:recusrion_two_mg}, we have that $V^*_{k}(\overline{x})$ is also concave in $\overline{x}$, and the power transfer control law~\eqref{eq:optimal_policy_delta_t} is optimal for all $k\in \{0,1,\ldots,K\}$. We remark that, as $\Delta t$ gets smaller ($\Delta t \to 0$) the region, $|x_1-x_2| \in (0, 2\overline{P}_{12}\Delta t]$, shrinks. Further, recall that,
\begin{align*}
    & \textstyle P(N,\mspace{-2mu}\alpha,\mspace{-2mu}\pi) \mspace{-4mu}=\mspace{-4mu} \textstyle \mathbb{P}[X_{1,T_f}^{sup}\mspace{-4mu}<\mspace{-4mu} N_B, \textstyle X_{1,T_f}^{inf}\mspace{-4mu}>\mspace{-4mu} 0,\textstyle \textstyle X_{2,T_f}^{sup}\mspace{-4mu}<\mspace{-4mu} N_B, X_{2,T_f}^{inf}\mspace{-4mu}>\mspace{-4mu} 0],\\
    & = \textstyle \mathbb{P}[X_1(t)\in[0,N_B], X_2(t)\in[0,N_B], \text{ for all }t\in [0,T_f]].
\end{align*}
On the other hand, by definition, $P_0^{\pi}(\overline{x}_0)=\mathbb{P}[\overline{X}_k\in [0,N_B], \text{ for all } k\in\{0,1,\ldots,K\}|\overline{X}_0=\overline{x}_0]$, with $\Delta t=\frac{T_f}{K}$. Thus taking $\Delta t \to 0$ we can show that the optimal power transfer policy, that maximizes $P(N,\alpha,\pi)$, is as follows:
\begin{align}\label{eq:optimal_policy_two_mg}
        \tilde{P}_{12}(t) = \begin{cases}
            \textstyle \overline{P}_{12}, & \text{ if } \textstyle X_1(t) > X_2(t),\\
            \textstyle -\overline{P}_{12}, & \text{ if } \textstyle X_1(t) < X_2(t),\\
            0, & \text{ if } \textstyle X_1(t)=X_2(t)
        \end{cases}
    \end{align}
\subsubsection{Battery Capacity Determination}
As per the optimal power transfer policy, at any time $t\in[0,T_f]$, the MG having higher available energy in its battery storage transfers power to the other one. Here, note that, $\tilde{P}_{12}(t)$ does not depend on the choice of $N$ and $\alpha$, thus, $\tilde{\pi}(N,\alpha)$ leads to the optimal solution to~\eqref{eq:problem_statement_two_mg_dual}; hence $\tilde{\pi}(N,\alpha)=\pi'$. Let $X_1^{\pi'}(t)$ and $X_2^{\pi'}(t)$ denote the total available energy in the BESs of MG-$1$ and MG-$2$ respectively, following the policy $\pi'$. Thus, from~\eqref{eq:brownian_soc_two_mg}, we have the following:
\begin{align}\label{eq:brownian_soc_two_mg_under_optimal_policy}
    \textstyle X_1^{\pi'}(t) &\mspace{-4mu}=\mspace{-4mu} \textstyle \alpha NB_{max} \mspace{-4mu}+\mspace{-4mu} \sigma W_1(t) \mspace{-4mu}-\mspace{-4mu} \int_0^t \tilde{P}_{12}(s)\text{d}s, \text{ and,} \nonumber\\
    \textstyle X_2^{\pi'}(t) &\mspace{-4mu}=\mspace{-4mu} \textstyle \alpha NB_{max} \mspace{-4mu}+\mspace{-4mu} \sigma W_2(t) \mspace{-4mu}+\mspace{-4mu} \int_0^t \tilde{P}_{12}(s)\text{d}s.
\end{align}
Further, from Theorem~\ref{theorem:problem_statement_two_mg_dual}, we have that, $\pi'$ solves problem~\eqref{eq:problem_statement_two_mg} as well. Hence, upon obtaining $\pi'$, we simplify problem~\eqref{eq:problem_statement_two_mg} as follows:
\begin{align}\label{eq:problem_statement_two_mg_1}
\min_{\alpha, N} & \textstyle \hspace{0.1in} N \\
  & \hspace{-1cm}\text{subject to} \ (\ref{eq:brownian_soc_two_mg_under_optimal_policy}),\ \textstyle N \geq 0,\ \textstyle 0 \leq \alpha \leq 1, \nonumber\\
  & \hspace{-1cm}\textstyle P(N,\alpha,\pi') \geq\mspace{-4mu} 1\mspace{-4mu}-\mspace{-4mu}\delta.\nonumber
\end{align}
In the subsequent analysis, we show how a tractable solution to the problem~(\ref{eq:problem_statement_two_mg_1}) is obtained to determine the number of BESs (or equivalently the BESs capacity of the MGs).

Let $X_{c}^{\pi'}(t):= X_1^{\pi'}(t)+X_2^{\pi'}(t)$ and $X_d^{\pi'}(t) := X_2^{\pi'}(t)-X_1^{\pi'}(t), t\in [0,T_f]$. Then, from~\eqref{eq:brownian_soc_two_mg_under_optimal_policy}, we have the following:
\begin{align}\label{eq:brownian_soc_two_mg_1}
    \textstyle X_c^{\pi'}(t) & \textstyle = 2\alpha NB_{max}\mspace{-4mu}+\mspace{-4mu}\sigma (W_1(t) + W_2(t))\nonumber\\
    \textstyle X_d^{\pi'}(t) & \textstyle = \sigma(W_2(t)-W_1(t))+2\int_0^t \tilde{P}_{12}(s)\text{d}s.
\end{align}
Further, let $X_{c,T_f}^{\pi',sup}:= \sup_{t\in[0,T_f]}X_c^{\pi'}(t)$ and $X_{c,T_f}^{\pi',inf}:= \inf_{t\in [0,T_f]}X_c^{\pi'}(t)$. Here we present the following proposition which provides an upper bound to the solution to~\eqref{eq:problem_statement_two_mg_1}:
\begin{proposition}\label{prop:conservative_sol_two_mg}
    Let $N^*$ denote the solution to~(\ref{eq:problem_statement_two_mg_1}) and $\beta \geq 0$ be a given constant. Further, consider the following problem:
    \begin{align}\label{eq:problem_statement_conservative_1_two_mg_1}
    \min_{\alpha, N} & \textstyle \hspace{0.1in} N \\
  & \text{subject to} \ (\ref{eq:brownian_soc_two_mg_1}),\ \textstyle N \geq 0,\ 0 \leq \alpha \leq 1, \nonumber\\
  & \textstyle \mathbb{P}[X_{c,T_f}^{\pi',sup} \geq 2NB_{max}-\beta]\textstyle  + \mathbb{P}[X_{c,T_f}^{\pi',inf} \leq \beta]\nonumber\\
  & \textstyle + \mathbb{P}[|X_{d,T_f}^{\pi',sup}| \geq \beta] \leq \delta.
\end{align}
Let $\hat{N}$ denote the optimal value of~(\ref{eq:problem_statement_conservative_1_two_mg_1}), then $N^*\mspace{-4mu}\leq\mspace{-4mu} \hat{N}$.
\end{proposition}
\begin{proof}
    Please see Appendix~\ref{appendix:conservative_sol_two_mg}.
\end{proof}
\noindent Note, $\mathbb{P}[|X_{d,T_f}^{\pi',sup}| \mspace{-4mu}\geq\mspace{-4mu} \beta]$ is monotonically decreasing in $\beta$. Hence, we choose $\beta$ in such a way that $\mathbb{P}[|X_{d,T_f}^{\pi',sup}| \mspace{-4mu}\geq\mspace{-4mu} \beta] \mspace{-4mu}<\mspace{-4mu} \frac{\delta}{2}$. Let $\delta' \mspace{-4mu}:=\mspace{-4mu} \frac{\delta}{2}$.
Thus, problem~\eqref{eq:problem_statement_conservative_1_two_mg_1} can be simplified as follows:
\begin{align}\label{eq:problem_statement_conservative_1_two_mg_2}
    \min_{\alpha, N} & \textstyle \hspace{0.1in} N \\
  & \text{subject to} \ (\ref{eq:brownian_soc_two_mg_1}),\ \textstyle N \geq 0,\ 0 \leq \alpha \leq 1, \nonumber\\
  & \textstyle \mathbb{P}[X_{c,T_f}^{\pi',sup} \mspace{-4mu}\geq\mspace{-4mu} 2NB_{max}\mspace{-4mu}-\mspace{-4mu}\beta]\textstyle  \mspace{-4mu}+\mspace{-4mu} \mathbb{P}[X_{c,T_f}^{\pi',inf} \leq \beta]\leq \delta'.
\end{align}
Let $(\alpha_c, N_c)$ solve the problem~\eqref{eq:problem_statement_conservative_1_two_mg_2}. Further note that, by definition \cite{shreve2004stochastic}, $(W_1(t)+W_2(t))\sim \sqrt{2}W_c(t)$, where $W_c(t)$ is a Weiner Process, and thus $X_c^{\pi'}(t)\sim \sigma_c W_c(t)+2\alpha NB_{max}$, where $\sigma_c = \sqrt{2}\sigma$. Thus, following the methodology described in Section~\ref{sec:solution_methodology_single_mg}, we can derive the following:
\begin{align}\label{eq:optimal_N_two_mg}
    N_c & \textstyle \leq \frac{\sqrt{8\sigma_c^2 T_f \ln(2/\delta')}}{2B_{max}} + \frac{\beta}{B_{max}},\nonumber\\
    \alpha_c & \textstyle = \frac{1}{2}.
\end{align}
Thus, we provide an upper bound to $N_c$ that solves~\eqref{eq:problem_statement_conservative_1_two_mg_2} which can be utilized by the DNO for practical purposes.

Next, we consider the effect of power line capacity on the battery capacity requirement of the MGs. We consider the following two extreme cases of the power line capacity: (i) $\overline{P}_{12}=0$, and (ii) $\overline{P}_{12}\to \infty$. First, consider the scenario when $\overline{P}_{12} = 0$, that is, the MGs are not connected to each other. In this case, each MG can be treated as a single MG, and the number of BESs required in each can be obtained from Theorem~\ref{theorem:solution_conservative_2_simplified}. Let $\tilde{N}_0$ denote the number of BESs required in each MG in this case. From Theorem~\ref{theorem:solution_conservative_2_simplified} (considering probability tolerance of $(\delta/2)$ for each MG in this case), we have, $\tilde{N}_{0} = c_0\sqrt{T_f}$, where, $c_0 := 2\sqrt{2}\sigma\sqrt{\ln({4/\delta})}/B_{max}$. Next, consider the scenario when the power transfer capacity is infinite. In this case, the combination of the two MGs can be treated as a single MG. Further, by definition \cite{shreve2004stochastic}, we have $(W_1(t)+W_2(t))\sim\sqrt{2}W_c(t) = \sigma_c W_c(t)$, where $W_c(t)$ is a Wiener process, and $\sigma_c=\sqrt{2}$. Here, the total energy available from the BESs, at time $t\in[0,T_f]$, is given by $X_1(t)+X_2(t)=2\alpha NB_{max}+\sigma (W_1(t)+W_2(t))$. Let $\tilde{N}_{\infty}$ denote the required number of BESs in each MG in this case. From Theorem~\ref{theorem:solution_conservative_2_simplified}, we have, $2\tilde{N}_{\infty} = 2\sqrt{2}\sigma_c \sqrt{\ln(2/\delta)}\sqrt{T_f}/B_{max}$, and hence $\tilde{N}_{\infty} = c_{\infty}\sqrt{T_f}$, where $c_{\infty}=2\sigma\sqrt{\ln(2/\delta)}/B_{max}$. Here, note that $\tilde{N}_{0} = \sqrt{\frac{2\ln(4/\delta)}{\ln(2/\delta)}}\tilde{N}_{\infty}$, thus, eventually, the required number of BESs decreases as the power flow capacity is increased up to a factor of $\sqrt{\frac{2\ln(4/\delta)}{\ln(2/\delta)}}$.
\section{Experimental Results}\label{sec:simu_results}
In this section, we demonstrate the applicability and utility of our proposed BES capacity design and power transfer policy determination methodology. First, we consider the case of a single microgrid where a modified IEEE-33 bus distribution network with a total load of $15$ kW is taken as the system under study. A battery, operated in Grid Forming Mode (GFM), is connected to bus $1$ and a wind farm is connected to bus $4$. The parameters, considered, for the real-time simulation, are as follows: $\sigma=1, \delta=0.02, B_{max}=1 \text{ kWh}, T_f=5$ hours. The power network is emulated inside OP5700 real-time simulator box, manufactured by OPAL-RT. The RT simulation is conducted using eMEGASIM-based EMT-simulation platform of RTLAB with a simulation step size of $50\mu s$. Using Theorem~\ref{theorem:solution_conservative_2_simplified}, GFM battery capacity is set to $13$ kWh with initial battery charge set to $6.5$ kWh. The system is run for $5$ hours in real-time. The wind and load variation and the power provided by the GFM battery are shown in Fig.~\ref{fig:power_5_zero_mu}. The battery State-of-the Charge (SoC) is shown in Fig.~\ref{fig:soc_5_zero_mu}. It can be seen that, throughout the time horizon of $5$ hours, the battery SoC neither hits the maximum capacity of $13$ kWh, nor gets completely depleted, establishing the utility of our proposed battery sizing methodology.

\noindent Next, we consider a system of two MGs, connected to each other via a power line. The parameters, considered for the simulation, are as follows: $\sigma=1, \delta=0.02, B_{max}=1 \text{ kWh}, \overline{P}_{12} = 15 \text{ kW}, T_f=5 \text{ hours}$. The time step, $\Delta t$, for the simulation is chosen to be $30$ seconds. Using ~\eqref{eq:optimal_N_two_mg}, the BES capacity of each MG is set to  $10$ kWh. The power transfer policy,~\eqref{eq:optimal_policy_delta_t}, is employed. Fig.~\ref{fig:SoC_two_mg} shows the simulation results for $5000$ runs, where we plot the available energy in the BES of each MG. It is observed that in only $0.4$\% of the total number of realizations, either the SoC of MG-1 battery or MG-2 battery hits the upper limit of $10$ kWh or hits the lower limit of $0$ kWh, which is well within the chosen limit of $2$\% ($\delta\mspace{-4mu}=\mspace{-4mu}0.02$). Next, we empirically analyze the effect of power line capacity, $\overline{P}_{12}$ under our proposed power transfer policy, on the required BES capacity to ensure that both the battery SoC remains within the operating range with a probability of at least $(1-\delta)=0.98$. Here, in Fig.~\ref{fig:P_12_vs_N}, we plot the required BES capacity, averaged over $5000$ runs, as we vary the power line capacity $\overline{P}_{12}$. It is observed that under our proposed power transfer policy, with the increase in line capacity, the MGs can share more power with each other, which in turn reduces the required BES capacity. Thus our proposed BES sizing and power transfer policy can provide a system design guideline for the distribution system operators. 
\begin{figure}[!t]
     \centering
     \subfloat{\includegraphics[scale=0.8,trim={0.1cm 0cm 0cm 0cm},clip]{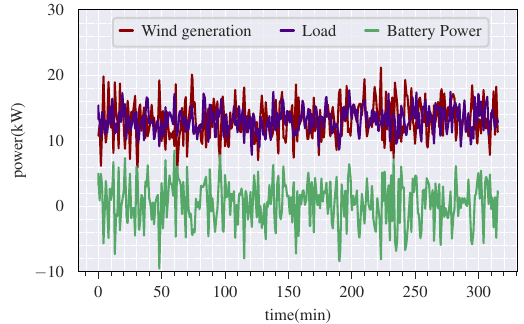}\label{fig:power_5_zero_mu}}\hspace{0.5cm}
     \subfloat{\includegraphics[scale=0.8,trim={0.1cm 0cm 0cm 0cm},clip]{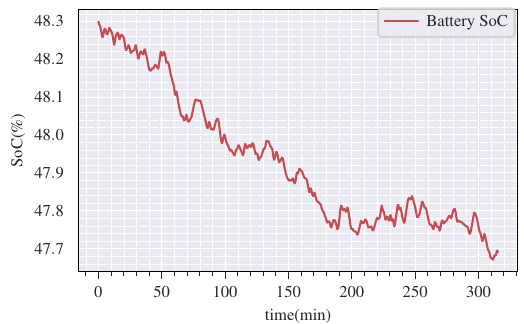}\label{fig:soc_5_zero_mu}}
     \caption{(a) Wind, load variation, and GFM battery power for $5$ hours in real-time in the IEEE-$33$ bus single MG system. (b) The battery does not get fully charged or get depleted throughout the time horizon.}
     \label{fig:simu_case_zero_mu}
     \vspace{-0.1cm}
\end{figure}
\begin{figure}[!t]
     \centering
     \includegraphics[scale=0.8,trim={0.1cm 0cm 0cm 0cm},clip]{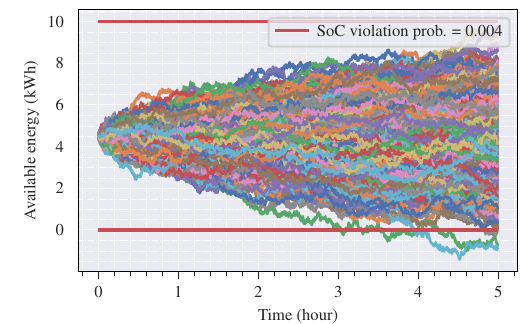}
     \caption{Two MG system: Battery capacity of each MG is $10$ kWh with a power line capacity of $15$ kW. Out of $5000$ realizations of the available energy in the batteries of MG-1 and MG-2, only 0.4\% realizations, hits either the upper limit of $10$ kWh or the lower limit of $0$ kWh.}
     \label{fig:SoC_two_mg}
     \vspace{-0.2cm}
\end{figure}
\begin{figure}[!t]
     \centering
     \includegraphics[scale=0.8,trim={0.1cm 0cm 0cm 0cm},clip]{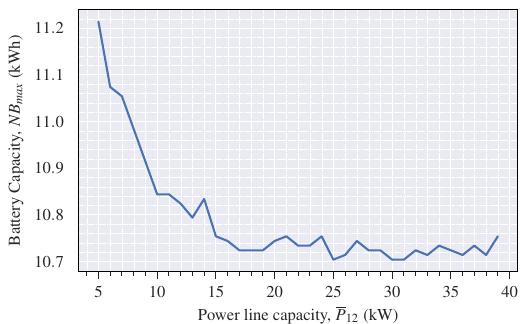}
     \caption{Required BES capacity for each MG decreases, under the optimal power transfer control policy, as power line capacity increases.}
     \label{fig:P_12_vs_N}
     \vspace{-0.5cm}
\end{figure}
\section{Conclusion}\label{sec:conclution}
In this article, we provide a BESs energy capacity determination methodology for microgrids that ensures the BESs do not reach their maximum capacity or get completely depleted throughout a chosen time horizon, with a pre-specified high probability. Such an operation guarantees that in case the RES generation is more than the demand, the BES can be charged and when demand exceeds the RES generation, the BESs can be utilized to supply the loads, thus reducing both generation and load curtailments. We further extend the methodology to a two MG system, where along with the required BES capacity, we provide a power transfer policy between the MGs that reduces the BESs requirement while ensuring the BESs remain within their operating range with high probability. We analytically show that BES requirement is minimized if at any time, the MG having higher available energy in its BESs, transfers maximum power that can flow via the power line, to the other. Further, as the power line capacity in increased the requied BES capacity can be decreased up to a factor of $\sqrt{\frac{2\ln(4/\delta)}{\ln(2/\delta)}}$.
\bibliography{reference}
\clearpage
\appendices
\section{Proof of Lemma~\ref{lem:gaussian_tail_bound}}\label{appendix:gaussian_tail_bound}
We have $X(t)=\sigma W(t) + \alpha NB_{max}$. Recall that for $a\in \mathbb{R}$, $T(a) = \inf\{t \geq 0 : X(t)=a\} = \inf \{t \geq 0: W(t)=\frac{a-\alpha NB_{max}}{\sigma}\}$. We define $\beta := \frac{a-\alpha NB_{max}}{\sigma}$. It follows from the properties of the Weiner process \cite{shreve2004stochastic}, that $T(NB_{max})$ is a stopping time.

\noindent Let $M(t) := e^{c W(t)-\frac{1}{2}c^2t}$ for any $c \neq 0$. Note that, the stochastic process $M(t)$ is a martingale \cite{karatzas2012brownian}. Therefore, for all $t\geq 0$, we have $\mathbb{E}[M(t)]=M(0)=1$. Further, we have $\mathbb{E}[M(T(a))]=1$ \cite{karatzas2012brownian}.

We first consider the case, $a=NB_{max}$. Here, we have,
\begin{align}\label{eq:martingale_soc}
    \textstyle & \mathbb{E}[M(T(NB_{max}))] = 1\nonumber\\
    \textstyle \implies & \mathbb{E}[e^{\frac{c}{\sigma}W(T(NB_{max}))-\frac{1}{2}c^2 T(NB_{max})}]\textstyle = 1\nonumber\\
    \textstyle \implies & \mathbb{E}[e^{\frac{c}{\sigma}X(T(NB_{max}))-\frac{1}{2}c^2T(NB_{max})}]\textstyle = e^{\frac{c \alpha N B_{max}}{\sigma}}.
\end{align}
Note that, in the case of $a=NB_{max}$, we have, $\beta=\frac{NB_{max}(1-\alpha)}{\sigma} \geq 0$. Therefore, we choose a positive constant $c$. Let $c=\sqrt{2s}$ for some $s \in \mathbb{R}^+$. Then, from~(\ref{eq:martingale_soc}), we have
\begin{align*}
    \textstyle & \mathbb{E}[e^{\frac{c}{\sigma}X(T(NB_{max}))-sT(NB_{max})}] = e^{\frac{c \alpha N B_{max}}{\sigma}}.
\end{align*}
By definition of $T(NB_{max})$, we obtain,
\begin{align*}
    \textstyle & \mathbb{E}[e^{\frac{c NB_{max}}{\sigma}-sT(NB_{max})}]=e^{\frac{c \alpha NB_{max}}{\sigma}}\nonumber\\
    & \textstyle \implies \mathbb{E}[e^{-sT(NB_{max})}] = e^{-\frac{\sqrt{2s}NB_{max}(1-\alpha)}{\sigma}}.
\end{align*}
Further, note that for any $s>0$, we have, $\mathbb{P}[T(NB_{max}) \leq T_f] = \mathbb{P}[e^{-sT(NB_{max})} \geq e^{-sT_f}]$. Thus, applying Markov's inequality we get,
\begin{align*}
    \textstyle \mathbb{P}[T(NB_{max}) & \textstyle \leq T_f] = \mathbb{P}[e^{-sT(NB_{max})} \geq e^{-sT_f}]\\
    & \textstyle \leq e^{sT_f}\mathbb{E}[e^{-sT(NB_{max})}]\\
    & \textstyle = e^{-\frac{\sqrt{2s}NB_{max}(1-\alpha)}{\sigma}+sT_f}
\end{align*}
Thus,
\begin{align}\label{eq:gaussian_tail_upper}
    \textstyle \mathbb{P}[T(NB_{max}) & \textstyle \leq T_f]\textstyle \leq \inf_{s>0} e^{-\frac{\sqrt{2s}NB_{max}(1-\alpha)}{\sigma}+sT_f}.
\end{align}
It can be shown that for any given $N\geq 0$ and $\alpha \in [0,1]$, the function $e^{-\frac{\sqrt{2s}NB_{max}(1-\alpha)}{\sigma}+sT_f}$ reaches global minima at $s^*=\frac{N^2B_{max}^2(1-\alpha)^2}{2\sigma^2 T_f}$. Plugging in the expression of $s^*$ in~(\ref{eq:gaussian_tail_upper}), we obtain the following:
\begin{align}\label{eq:gaussian_tail_upper_final}
    \textstyle \mathbb{P}[T(NB_{max}) & \textstyle \leq T_f]\textstyle \leq e^{-\frac{N^2 B_{max}^2 (1-\alpha)^2}{2\sigma^2 T_f}}.
\end{align}
Carrying out a similar procedure for the case $a=0$, to obtain an upper bound for $\mathbb{P}[T(0) \leq T_f]$, we can derive the following:
\begin{align}\label{eq:gaussian_tail_lower_final}
    \textstyle \mathbb{P}[T(0) & \textstyle \leq T_f] \leq e^{-\frac{\alpha^2 N^2 B_{max}^2}{2\sigma^2 T_f}}.
\end{align}
Combining both~(\ref{eq:gaussian_tail_upper_final}) and~(\ref{eq:gaussian_tail_lower_final}), we obtain the desired result. This completes the proof.
\section{Proof of Theorem~\ref{theorem:solution_conservative_2_simplified}}\label{appendix:solution_conservative_2_simplified}
Recall that 
\[f(N, \alpha) := e^{-\frac{N^2B_{max}^2(1-\alpha)^2}{2\sigma^2 T_f}} + e^{-\frac{N^2B_{max}^2\alpha^2}{2\sigma^2 T_f}}.\]
For $\alpha=\frac{1}{2}$, we have $f(N,\frac{1}{2})=2e^{-\frac{N^2B_{max}^2}{8\sigma^2 T_f}}$.
Note that, $f(N,\frac{1}{2})$ is strictly monotonic decreasing in $N$, for $N > 0$. Thus, for $\alpha=\frac{1}{2}$, the solution to~(\ref{eq:problem_statement_conservative_2_1}) is obtained when the constraint $f(N,\frac{1}{2})\leq \delta $ is an equality. Solving $f(N,\frac{1}{2})=\delta$  we obtain $N(\frac{1}{2})=N_o=\frac{\sqrt{8\sigma^2 T_f \ln(2/\delta)}}{B_{max}}$.

We define $k:=\frac{2}{\delta}$. Thus, for any $\alpha \in (0,1)$, we have,
\begin{align*}
    \textstyle f(N_0, \alpha) = \textstyle e^{-4\ln(k)(1-\alpha)^2} + e^{-4\ln(k)\alpha^2}.
\end{align*}
We first prove that $f(N_0,\alpha)$ achieves global minima at $\alpha=\frac{1}{2}$. Differentiating with respect to $\alpha$, we have,
\begin{align}\label{eq:partial_f_alpha}
    \hspace{-0.3cm} \textstyle \frac{\partial f(N_0,\alpha)}{\partial \alpha} & \mspace{-4mu}=\mspace{-4mu} \textstyle 8(1\mspace{-4mu}-\mspace{-4mu}\alpha)\ln(k)(k)^{-4(1-\alpha)^2}\mspace{-8mu}-\mspace{-4mu}8\alpha \ln(k)(k)^{-4\alpha^2}.
\end{align}
Note, for any $\alpha \in (0,1)$ that satisfies $\frac{\partial f(N_0,\alpha)}{\partial \alpha}=0$, we have the following:
\begin{align}\label{eq:partial_f_alpha_simplified}
    \textstyle (k)^{4(1-2\alpha)} &= \textstyle \frac{1-\alpha}{\alpha}.
\end{align} 
Note that $k=\frac{2}{\delta}\geq 2.$ Clearly, $\alpha=\frac{1}{2}$ satisfies~(\ref{eq:partial_f_alpha_simplified}). Double differentiation of $f(N_0,\alpha)$ with respect to $\alpha$ leads to the following:
\begin{align*}
    \textstyle \frac{\partial^2 f(N_0,\alpha)}{\partial \alpha^2} & \mspace{-4mu}=\mspace{-4mu} \textstyle (8(1\mspace{-4mu}-\mspace{-4mu}\alpha)\ln(k))^2(k)^{-4(1\mspace{-2mu}-\mspace{-2mu}\alpha)^2}\mspace{-8mu}\\
    & \textstyle -\mspace{-4mu}8\ln(k)(k)^{-4(1\mspace{-2mu}-\mspace{-2mu}\alpha)^2}\\
    & \textstyle +(8\alpha \ln(k))^2(k)^{-4\alpha^2}-8\ln(k)(k)^{-4\alpha^2}.
\end{align*}
For any $\alpha$ that satisfies $\frac{\partial f(N_0,\alpha)}{\partial \alpha}=0$, from~(\ref{eq:partial_f_alpha}), we get,
\begin{align}\label{eq:partial_f_alpha_double}
    & \textstyle \frac{\partial^2 f(N_0,\alpha)}{\partial \alpha^2} \mspace{-4mu}=\mspace{-4mu} \textstyle (8\alpha \ln(k))^2(k)^{8(1-\alpha)^2-8\alpha^2}(k)^{-4(1-\alpha)^2}\nonumber\\
    & \textstyle -8\ln(k)(k)^{-4(1-\alpha)^2}+(8\alpha \ln(k))^2(k)^{-4\alpha^2}-8\ln(k)(k)^{-4\alpha^2}\nonumber\\
    &= \textstyle (8\alpha\ln(k))^2(k)^{-4\alpha^2}(1+(k)^{4(1-\alpha)^2-4\alpha^2})\nonumber\\
    & - \textstyle 8 \ln(k)(k)^{-4\alpha^2}(1+(k)^{4\alpha^2-4(1-\alpha)^2})\nonumber\\
    &= \textstyle (8\alpha\ln(k))^2(k)^{-4\alpha^2}(1+(k)^{4(1-2\alpha)})\nonumber\\
    & \textstyle -8 \ln(k)(k)^{-4\alpha^2}(1+(k)^{4(2\alpha-1)})\nonumber\\
    & = \textstyle (8\alpha\ln(k))^2(k)^{-4\alpha^2}(1+\frac{1-\alpha}{\alpha})-8 \ln(k)(k)^{-4\alpha^2}(1+\frac{\alpha}{1-\alpha})\nonumber\\
    &= \textstyle (8\ln(k))^2(k)^{-4\alpha^2}\alpha-8 \ln(k)(k)^{-4\alpha^2}(\frac{1}{1-\alpha})\nonumber\\
    &= \textstyle 8\ln(k)(k)^{-4\alpha^2}(8\ln(k)\alpha-\frac{1}{1-\alpha}).
\end{align}
Note that, for $\alpha=\frac{1}{2}$,~(\ref{eq:partial_f_alpha_double}) can be simplified to $16\ln(k)(\frac{1}{k})(2\ln(k)-1)$ which is positive, since $k=\frac{2}{\delta} > 2$. Thus, $f(N_0,\alpha)$ achieves minima at $\alpha=\frac{1}{2}$.

Next, we show that $f(N_0,\alpha)$ achieves global minima at $\alpha=\frac{1}{2}$, that is for $\alpha \neq \frac{1}{2},\alpha \in (0,1)$, we have $\frac{\partial^2 f(N_0,\alpha)}{\partial \alpha^2} < 0$, or equivalently, from~\eqref{eq:partial_f_alpha_double}, $(8\ln(k)\alpha-\frac{1}{1-\alpha}) < 0$. 
We first define $g(\alpha) := 2\ln(\frac{1}{\alpha})-2\ln(\frac{1}{1-\alpha})-\frac{1}{\alpha}+\frac{1}{1-\alpha}$. Here we get, $\frac{\partial g(\alpha)}{\partial \alpha}=(1-2\alpha)(\frac{1}{\alpha^2}-\frac{1}{(1-\alpha)^2})$.
In the region, $\alpha \in (0,\frac{1}{2})$, we have, $0<\frac{1}{1-\alpha}<\frac{1}{\alpha}$ and $(1-2\alpha)>0$. Thus, $\frac{\partial g(\alpha)}{\partial \alpha}>0$, hence, $g(\alpha)$ is strictly increasing in $\alpha\in(0,\frac{1}{2})$. Similarly, in the region $\alpha\in (\frac{1}{2},1)$, we have $(1-2\alpha)<0$, and $0<\frac{1}{\alpha}<\frac{1}{1-\alpha}$. Thus, $\frac{\partial g(\alpha)}{\partial \alpha}>0$, hence $g(\alpha)$ is increasing in $\alpha \in (\frac{1}{2},1)$. Since, $g(\alpha)$ is a continuous function of $\alpha\in(0,1)$, we have
\begin{align}\label{eq:partial_g_alpha}
    \textstyle g(\alpha)& \textstyle <g(\frac{1}{2})=0, \ \alpha \in (0,\frac{1}{2}),\nonumber\\
    \textstyle g(\alpha) & \textstyle > g(\frac{1}{2})=0,\ \alpha \in (\frac{1}{2},1).
\end{align}
Next, we will show that $(8\ln(k)\alpha - \frac{1}{1-\alpha})<0, \alpha\in(0,1)\setminus \{\frac{1}{2}\}$. First consider the case $\alpha \in (0,\frac{1}{2})$. From~\eqref{eq:partial_f_alpha_simplified}, we have $8\ln(k)=\frac{2}{1-2\alpha}\ln(\frac{1-\alpha}{\alpha})$, and since $\alpha\in(0,\frac{1}{2})$, we have $(1-2\alpha)>0$. Now, suppose, $(8\ln(k)\alpha-\frac{1}{1-\alpha}) \geq 0$. Here we have,
\begin{align*}
    & \textstyle 8\ln(k)\alpha-\frac{1}{1-\alpha} \geq 0\\
    \implies & \textstyle \frac{2\alpha}{1-2\alpha}\ln(\frac{1-\alpha}{\alpha})-\frac{1}{1-\alpha}\geq 0\\
    \implies & \textstyle 2\ln(\frac{1-\alpha}{\alpha})-\frac{1-2\alpha}{\alpha(1-\alpha)}\geq 0\\
    \implies & \textstyle 2\ln(\frac{1}{\alpha})-2\ln(\frac{1}{1-\alpha})-\frac{1}{\alpha}+\frac{1}{1-\alpha}\geq 0\\
    \implies & \textstyle g(\alpha)\geq 0.
\end{align*}
However, from~\eqref{eq:partial_g_alpha}, we have $g(\alpha)<0$ when $\alpha\in(0,\frac{1}{2})$, and thus, $(8\ln(k)\alpha-\frac{1}{1-\alpha})<0$.

Next, consider the region $\alpha\in(\frac{1}{2},1)$. Here, we have $(1-2\alpha)<0$. Suppose $(8\ln(k)\alpha-\frac{1}{1-\alpha})\geq 0$. Here, we have,
\begin{align*}
    & \textstyle 8\ln(k)\alpha-\frac{1}{1-\alpha} \geq 0\\
    \implies & \textstyle \frac{2\alpha}{1-2\alpha}\ln(\frac{1-\alpha}{\alpha})-\frac{1}{1-\alpha}\geq 0\\
    \implies & \textstyle 2\ln(\frac{1-\alpha}{\alpha})-\frac{1-2\alpha}{\alpha(1-\alpha)}\leq 0\\
    \implies & \textstyle 2\ln(\frac{1}{\alpha})-2\ln(\frac{1}{1-\alpha})-\frac{1}{\alpha}+\frac{1}{1-\alpha}\leq 0\\
    \implies & \textstyle g(\alpha)\leq 0.
\end{align*}
However, from~\eqref{eq:partial_g_alpha}, we have $g(\alpha)>0$ when $\alpha\in(\frac{1}{2},1)$, and thus, $(8\ln(k)\alpha-\frac{1}{1-\alpha})<0$.
Therefore, for any $\alpha\in (0,1)\setminus \{\frac{1}{2}\}$ that satisfies~(\ref{eq:partial_f_alpha_simplified}), $\frac{\partial^2f(N_0,\alpha)}{\partial \alpha^2} < 0$. Thus, $f(N_0,\alpha)$ achieves global minima at $\alpha=\frac{1}{2}$. Hence, for any $\alpha \neq \frac{1}{2}$, we have $f(N_0,\alpha) > f(N_0,\frac{1}{2})=2e^{-\ln(k)}=\delta$. Clearly, for an $\alpha\in (0,1)\setminus \{\frac{1}{2}\}$, the tuple $(\alpha, N_o)$ is not a feasible solution to~\eqref{eq:problem_statement_conservative_2_1}. Further, note that, $f(N,\alpha)$ is strictly monotonic decreasing in $N$. Thus for any $\alpha\in (0,1)\setminus \{\frac{1}{2}\}$, $M\leq N_o,\ f(M,\alpha)\geq f(N_o,\alpha)> \delta$; thus such a $M$ remains infeasible for $\alpha\not=\frac{1}{2}.$  Hence, for any $\alpha\in (0,1)\setminus \{\frac{1}{2}\}$, we must have $N(\alpha) \geq N(\frac{1}{2})$.
This completes the proof.
\section{Proof of Theorem~\ref{theorem:problem_statement_two_mg_dual}}\label{appendix:problem_statement_two_mg_dual}
Note that as $(N^*,\alpha^*,\pi^*)$ solves~\eqref{eq:problem_statement_two_mg}, we have $(N^*,\alpha^*,\pi^*)\mspace{-4mu}\in\mspace{-4mu} S_1 \mspace{-4mu}\subseteq\mspace{-4mu} F_1$, and thus~\eqref{eq:brownian_soc_two_mg} is satisfied, $\pi^* \in \Pi$, and $P(N^*,\alpha^*,\pi^*)\geq 1-\delta$. Further, as $N^*\leq N^*$, it follows that $(N^*,\pi^*)\mspace{-4mu}\in\mspace{-4mu} F_2$ as well and thus $P(N^*,\pi^*,\alpha^*)\leq \nu$, as $\nu$ is the optimal value of the maximization problem~\eqref{eq:problem_statement_two_mg_dual}. Further, as 
$(N',\pi')$ solves~\eqref{eq:problem_statement_two_mg_dual}, we have $P(N',\alpha^*,\pi')=\nu\geq P(N^*,\alpha^*,\pi^*)\geq 1-\delta$. Thus, in summary, we have $P(N',\alpha^*,\pi')\geq 1-\delta$ and as $(N',\pi')\mspace{-4mu}\in\mspace{-4mu} S_2\mspace{-4mu}\subseteq\mspace{-4mu}F_2$ $\pi'\in \Pi, N'\leq N^*$, and~\eqref{eq:brownian_soc_two_mg} is satisfied. Hence, $(N',\alpha^*,\pi')\mspace{-4mu}\in\mspace{-4mu}F_1$ as well.

\noindent Note that, as~\eqref{eq:problem_statement_two_mg} is a minimization problem, it follows that $N'\geq N^*$. However, as $(N',\pi')\mspace{-4mu}\in\mspace{-4mu}S_2\mspace{-4mu}\subseteq\mspace{-4mu}F_2$, we get $N'\leq N^*$. Thus $N'=N^*$.

\noindent We have established that, $(N',\alpha^*,\pi')\in F_1$, and $N'=N^*$, thus $(N',\alpha,\pi')$ solves~\eqref{eq:problem_statement_two_mg} as well, that is $(N',\alpha^*,\pi')\in S_1$. This completes the proof.
\section{Proof of Proposition~\ref{prop:conservative_sol_two_mg}}\label{appendix:conservative_sol_two_mg}
Let $A_1 := \{X_{1,T_f}^{\pi',sup}(t) \geq NB_{max}\}, A_2 := \{X_{1,T_f}^{\pi',inf}(t)\leq 0\}, A_3:= \{X_{2,T_f}^{\pi',sup}(t) \geq NB_{max}\}$ and $A_4 := \{X_{2,T_f}^{\pi',inf}(t)\leq 0\}$. Therefore, the chance-constraint of the problem~(\ref{eq:problem_statement_two_mg_1}) can be equivalently written as:
\begin{align*}
    \textstyle & \textstyle \mathbb{P}[A_1 \cup A_2 \cup A_3 \cup A_4] \leq \delta.
\end{align*}
In the subsequent paragraphs, we establish that a feasible solution to~(\ref{eq:problem_statement_conservative_1_two_mg_1}) is also a feasible solution to~(\ref{eq:problem_statement_two_mg_1}).

Consider the given constant $\beta$, and let $A_d := \{|X_{d,T_f}^{\pi',sup}| < \beta$ and it's complement, $A_d' := \{|X_{d,T_f}^{\pi',sup}| \geq \beta\}$. Marginalizing over the event sets $A_d$ and $A_d'$, we derive the following:
\begin{align}\label{eq:prob_marginalization_two_mg}
    \textstyle \mathbb{P}[A_1 \cup A_2 \cup A_3 \cup A_4]& \textstyle= \textstyle \mathbb{P}[(A_1 \mspace{-2mu}\cup\mspace{-2mu} A_2 \mspace{-2mu}\cup\mspace{-2mu} A_3 \mspace{-2mu}\cup\mspace{-2mu} A_4) \mspace{-2mu}\cap\mspace{-2mu} A_d] \nonumber\\
    & \textstyle \mspace{-4mu}+\mspace{-4mu} \mathbb{P}[(A_1 \mspace{-2mu}\cup\mspace{-2mu} A_2 \mspace{-2mu}\cup\mspace{-2mu} A_3 \mspace{-2mu}\cup\mspace{-2mu} A_4) \mspace{-2mu}\cap\mspace{-2mu} A_d'].
\end{align}
First, we focus on the first term on the right-hand side of~(\ref{eq:prob_marginalization_two_mg}), and derive the following:
\begin{align}\label{eq:prob_marginalization_two_mg_first}
    & \textstyle \mathbb{P}[(A_1 \cup A_2 \cup A_3 \cup A_4) \cap A_d]\nonumber\\
    & \textstyle \leq \mathbb{P}[(A_1 \cup A_3) \cap A_d]\textstyle + \mathbb{P}[(A_2 \cup A_4) \cap A_d].
\end{align}
Now, let $\hat{A}_d := \{|X_{d}^{\pi'}(t)|< \beta, \forall t\in[0,T_f]\}$. As, $|X_{d,T_f}^{\pi',sup}|< \beta \implies |X_{d}^{\pi'}(t)|< \beta$, for all $t \in [0,T_f]$, we have, $\mathbb{P}[A_d] \leq \mathbb{P}[\hat{A}_d]$. Thus, $\mathbb{P}[(A_1 \cup A_3) \cap A_d] \leq \mathbb{P}[(A_1 \cup A_3) \cap \hat{A}_d] = \mathbb{P}[(A_1 \cap \hat{A}_d) \cup (A_3 \cap \hat{A}_d)]$. Now, let $A_c := \{X_{c,T_f}^{\pi',sup} \geq 2NB_{max}-\beta\}$. It can be shown that $(A_1 \cap \hat{A}_d) \subseteq A_c$ and $(A_3 \cap \hat{A}_d) \subseteq A_c$, which leads to $\mathbb{P}[(A_1 \cap \hat{A}_d) \cup (A_3 \cap \hat{A}_d)] \leq \mathbb{P}[A_c]$. Thus, we have
\begin{align}\label{eq:prob_marginalization_two_mg_first_1}
    & \textstyle \mathbb{P}[(A_1 \cup A_3) \cap A_d] \leq \mathbb{P}[A_c].
\end{align}
In a similar manner, we define $\tilde{A}_c := \{X_{c,T_f}^{\pi',inf} < \beta\}$, and derive the following:
\begin{align}\label{eq:prob_marginalization_two_mg_first_2}
    & \textstyle \mathbb{P}[(A_2 \cup A_4) \cap A_d] \leq \mathbb{P}[\tilde{A}_c].
\end{align}
Hence, from~(\ref{eq:prob_marginalization_two_mg_first}) and combining~(\ref{eq:prob_marginalization_two_mg_first_1}) and~(\ref{eq:prob_marginalization_two_mg_first_2}), we have
\begin{align}\label{eq:prob_marginalization_two_mg_first_upper_bound}
    & \textstyle \mathbb{P}[(A_1 \mspace{-2mu}\cup\mspace{-2mu} A_2 \mspace{-2mu}\cup\mspace{-2mu} A_3 \mspace{-2mu}\cup\mspace{-2mu} A_4) \mspace{-2mu}\cap\mspace{-2mu} A_d] \textstyle \leq \mathbb{P}[A_c] \textstyle + \mathbb{P}[\tilde{A}_c].
\end{align}
Thus, from~(\ref{eq:prob_marginalization_two_mg_first_upper_bound}) and~(\ref{eq:prob_marginalization_two_mg}), we have
\begin{align}\label{eq:prob_marginalization_two_mg_upper_bound}
    & \textstyle \mathbb{P}[A_1 \cup A_2 \cup A_3 \cup A_4] \nonumber\\
    & \textstyle \leq \mathbb{P}[A_c] \textstyle+ \mathbb{P}[\tilde{A}_c] \textstyle + \mathbb\mathbb{P}[(A_1 \cup A_2 \cup A_3 \cup A_4)\cap A_d']\nonumber\\
    & \textstyle \leq \mathbb{P}[A_c] \textstyle+ \mathbb{P}[\tilde{A}_c] \textstyle + \mathbb\mathbb{P}[A_d'].
\end{align}
Note that the chance-constraint of the problem~(\ref{eq:problem_statement_conservative_1_two_mg_1}) can be written as $\mathbb{P}[A_c]+\mathbb{P}[\tilde{A}_c]+\mathbb\mathbb{P}[A_d'] \leq \delta$.
Thus, we establish that a feasible solution to~(\ref{eq:problem_statement_conservative_1_two_mg_1}) is also a feasible solution to~(\ref{eq:problem_statement_two_mg_1}). Thus $N^* \leq \hat{N}$. This completes the proof.
\end{document}